\documentclass[10pt]{article}
\usepackage{appendix}
\usepackage{setspace}
\usepackage[font=small,  margin=2cm]{caption}
\usepackage[tbtags]{amsmath}
\usepackage{amsthm,amssymb,amsfonts}
\usepackage[numbers]{natbib}
\usepackage{multirow}
\usepackage{subcaption}
\usepackage{lscape}
\usepackage{makecell}
\usepackage{exscale}
\usepackage{booktabs}
\usepackage{array}
\usepackage{fullpage}
\usepackage{url}
\usepackage{algorithm}
\usepackage{algpseudocode}
\usepackage{bm}
\usepackage{smile}
\usepackage{mathtools}
\usepackage{wrapfig}
\usepackage{lipsum}
\usepackage{mathrsfs}
\usepackage{relsize}
\usepackage{dsfont}
\usepackage{multirow}
\usepackage{apalike}
\usepackage{chngcntr}
\usepackage[usenames,dvipsnames,svgnames,table]{xcolor}

\numberwithin{equation}{section}
\numberwithin{theorem}{section}
\numberwithin{corollary}{section}
\counterwithout{asmp}{section}
\numberwithin{definition}{section}

\begin{document}

\title{\LARGE Network-Assisted Estimation for Large-dimensional Factor Model \\ with Guaranteed Convergence Rate Improvement}

	\author{Long Yu\thanks{ School of Management, Fudan University, Shanghai, China;  Email:{\tt fduyulong@163.com}.},~~Yong He\thanks{ School of Statistics, Shandong University of Finance and Economics, Jinan, China; Email:{\tt heyong@sdufe.edu.cn}.},~~ Xinsheng Zhang\thanks{ School of Management, Fudan University, Shanghai, China; Email:{\tt xszhang@fudan.edu.cn}.}, ~~ Ji Zhu \thanks{Department of Statistics, University of Michigan, Ann Arbor,
Michigan, U.S.A.; Email:{\tt jizhu@umich.edu}}}	
	\date{}	
	\maketitle
Network structure is growing popular for capturing the intrinsic relationship between  large-scale variables. In the paper we propose to improve the estimation accuracy for large-dimensional  factor model when a network structure between individuals is observed. To fully excavate the prior network information, we construct two different penalties to regularize the factor loadings and shrink the idiosyncratic errors.  Closed-form solutions are provided for the penalized optimization problems. Theoretical results  demonstrate that the modified estimators achieve faster convergence rates and lower asymptotic mean squared errors when the underlying network structure among individuals is correct. An interesting finding is that even if the priori network  is totally misleading, the proposed estimators perform no worse than conventional state-of-art methods.  Furthermore, to facilitate the practical application, we propose a data-driven approach to select the tuning parameters, which is  computationally efficient. We also provide an empirical criterion to determine the number of common factors. Simulation studies and application to the S\&P100 weekly return dataset convincingly illustrate the superiority and adaptivity of the new approach.

\vspace{2em}


\section{Introduction}
Factor models are widely used to extract representative features and  explain the generative process of massive variables. They have been successfully applied on financial engineering, economic analysis and biological technology, e.g., to study the expected returns (\cite{ross1976arbitrage,ross1977capital}, \cite{fama1992cross,fama1993common}) and risks  of portfolios (\cite{fan2015risks}), to characterize macroeconomic observations (\cite{stock2002forecasting}) and to analyse gene expression data (\cite{mayrink2013sparse}).  On the other hand,  after adjusted with a factor structure, the  efficiency of many statistical learning  methods can be improved, such as factor profiled sure independence screening in \cite{wang2012factor}, the estimation of  large covariance matrices in \cite{fan2018large} and factor-adjusted multiple testing in \cite{fan2019farmtest}.

With diversified assumptions, the forms of factor models vary a lot in the literature. In this paper we primarily focus on the  approximate factor structure  proposed by \cite{chamberlain1983arbitrage}, which allows the idiosyncratic errors to be cross-sectionally correlated. The model is further studied in \cite{bai2002determining} and becomes a popular scheme for modern factor analysis. A fundamental problem for factor analysis is to accurately specify the latent factor scores and loadings, which becomes challenging due to  the large-dimensionality and the existence of both serial and cross sectional dependencies. \cite{bai2003inferential} studied the consistency and asymptotic normality of the estimators by principal component analysis (PCA), while the properties of maximum likelihood estimators are explored in \cite{bai2012statistical} and \cite{bai2016maximum}. Certain recent works consider more flexible factor model settings, for instance,  \cite{bates2013consistent} provided consistent estimator for factor models with potential structural breaks of the factor loadings, which can be further applied to localize the structural breaks in the model setting by \cite{baltagi2017identification}.

Though there are extensive works on the estimation of factor models, the majority of them usually turn a blind eye for the existence of multiple-sources data. For example, when explaining the expected returns of portfolios, auxiliary information such as capital size and  the connections of corresponding public companies shall be quite useful.  Actually, it's becoming more and more important to study the combinations of diversified data to embrace the ``Big Data'' era. Unfortunately, only few works dedicate to consider factor models with data from multiple sources . The constrained factor model in \cite{tsai2010constrained} is an example, which imposes completely or partially linear constraints on factor loadings. The constraints are derived from industrial sectors in their real data  application on the excess stock returns of 10 U.S. companies. The idea is extended to doubly constraints in \cite{tsai2016doubly} and  matrix-value factor structure in \cite{chen2019constrained}. As another example, the semi-parametric factor structure in \cite{connor2012efficient} and \cite{fan2016projected} models the factor loadings as unknown functions of relevant covariates such as market capitalization and price-earning ratio. Factor models with auxiliary information deserve more attention while the potential abuse of such information shall be also considered, which makes the problem even more challenging.

We propose a  network-linked  factor model as an example to incorporate the auxiliary information from multiple-sources data in this paper. Network analysis has been popular for studying complex systems, such as the co-authorship network \citep{ji2016coauthorship} and gene-interaction network \citep{liu2017structure,xia2018multiple}. Taking network cohesion into account can lead to improvements of traditional methods, see, for example, the regression model in \cite{yu2016sparse}, the prediction models in \cite{li2019prediction}, the classification model in \cite{liu2019graph}. Yet as far as we know, network-assisted estimation of loadings and  scores in factor models has not been studied.   In the current article, a variable-to-variable network is assumed to be observed  in the factor model, and represents intrinsic similarities of the loadings. The assumption quite makes sense in real application, for instance, in financial markets the network can always be observed according to  the industrial sectors or ownership relations of public companies. We aim to improve the estimation accuracy for high dimensional  factor models with the prior network information, while the new approach shall also perform robustly even if  the observed  network is actually irrelevant.

We summarize the major contributions of this paper as follows. Firstly, we establish the framework which incorporates priori network  information into factor models. The framework is natural, interpretable, and can be flexibly extended to factor-based statistical learning methods. Secondly, we incorporate the network information by exerting two different penalties. An interesting finding is that the relationship between the two penalized problems are analogous to those between the ridge regression and principal component regression. Thirdly, we provide closed-form solutions to the penalized problems and  unified theoretical analysis for the corresponding estimators. Theoretical results show that the proposed methods achieve  guaranteed convergence rate improvement when priori network is correct while performs no worse than conventional PCA method even if the priori network  is totally misleading. Finally, we propose a data-driven $C_L$ criterion to select the tuning parameters for the estimation. The criterion is computationally  more  efficient than generic cross validation (CV), and  can be extended to deal with more general factor models where tuning parameters are involved.  We also  provide an adaptive method to specify the number of common factors with the observed network.

The notations throughout this paper are introduced here. For a matrix $\Ab$, $\mathrm A_{ij}$ or $\mathrm A_{i,j}$ is the $i$-th row $j$-th column entry of $\Ab$, tr($\Ab$) denotes the trace of $\Ab$ and  $(\mathrm A_{ij})_{n_1\times n_2}$ is a $n_1\times n_2$ matrix with entries $\mathrm A_{ij}$. We denote $\lambda_j(\Ab)$ as the $j$-th eigenvalue of a  symmetric matrix $\Ab$ while $\lambda_{\max}(\Ab)$ and $\lambda_{\min}(\Ab)$ correspond to the maximum and minimum eigenvalues. We denote $\|\Ab\|_F=\text{tr}^{1/2}(\Ab^\top\Ab)$, $\|\Ab\|=\lambda_{\max}^{1/2}(\Ab^\top\Ab)$ as the Frobenius norm and spectral norm of $\Ab$ respectively. For a vector $\bv$, $\|\bv\|$ is denoted as the Euclidean norm of $\bv$. We define $\text{diag}(a_1,\ldots,a_p)$ as the diagonal matrix with $j$-th diagonal entry being $a_j$ and the definition also holds if the diagonal entries are square matrices. $X_n\gtrsim Y_n$ means there exists constant $c>0$ such that $X_n\ge cY_n$ for  sufficiently large $n$, while $X_n\lesssim Y_n$ means there exists $c>0$ such that $X_n\le cY_n$ for sufficiently large $n$, and $X_n\asymp Y_n$ means both $X_n\gtrsim Y_n$ and $X_n\lesssim Y_n$ hold. For a random variable (or vector) $\bX$, ${\rm E}(\bX)$ denotes the expectation. The constants $c$ and $M$ in different lines can be not identical.

The  rest of the paper is structured as follows. Section \ref{sec2} illustrates how to fuse an observed network into the  factor models. Section {\ref{sec3}} displays the penalized estimation procedures, while theoretical assumptions and analysis are presented in section \ref{sec4}. In section \ref{sec5}, we focus on the selection of tuning parameters, including the number of common factors. Simulation studies and real data application are organized separately in section \ref{sec6} and \ref{sec7}, to empirically compare the proposed approaches with conventional PCA solution. Section \ref{sec8} concludes the paper and discusses the future works while all the technical proofs are in our supplementary materials..

\section{Model Setting}\label{sec2}
\subsection{Preliminary}
This paper focuses on the static approximate factor model originated from \cite{chamberlain1983arbitrage} and further studied in \cite{bai2002determining}. It assumes that the volatilities  of large-scale variables are driven by some latent  common factors and idiosyncratic errors,  which can be generally written as
\begin{equation}\label{equ2.1}
x_{tj}=\sum_{k=1}^{r}b_{jk}f_{tk}+e_{tj},\quad\text{for}\quad t=1,\ldots,T, j=1,\ldots,p,
\end{equation}
where the idiosyncratic errors can be cross-sectionally correlated. In model (\ref{equ2.1}), only $x_{tj}$ are observable, while the number of common factors $r$ (assumed to be fixed), the factor loadings $b_{jk}$, the random factor scores $f_{tk}$ and  idiosyncratic  errors $e_{tj}$ are all unobserved. The primary research interest of this paper is to estimate factor loadings, factor scores and further the common components $c_{tj}=\sum_{k=1}^{r}b_{jk}f_{tk}$ for the large-dimensional case $p\gtrsim T$.

 Model (\ref{equ2.1}) is also frequently written in vector form or matrix form as
\begin{equation}
\bx_t=\Bb\bF_t+\be_t,\quad\text{or}\quad\Xb=\Fb\Bb^\top+\Eb,
\end{equation}
where $\bx_t^\top=(x_{t1},\ldots,x_{tp})$, $\Bb=(b_{jk})_{p\times r}$, $\bF_t^\top=(f_{t1},\ldots,f_{tr})$, $\be_t^\top=(e_{t1},\ldots,e_{tp})$, $\Xb=(x_{tj})_{T\times p}$, $\Fb=(f_{tk})_{T\times r}$ and $\Eb=(e_{tj})_{T\times p}$. By the matrix form, it's easy to see the factor model is not identifiable unless exerting additional constrains. Most existing literatures  assume $\Fb^\top\Fb/T=\Ib_r$ almost surely as $T\rightarrow\infty$ to ensure identifiability up to orthogonal transformations. More detailed discussion about identifiability can be found in \cite{bai2012statistical} and  \cite{bai2013principal}. For  model (\ref{equ2.1}), principal component analysis (PCA) and maximum likelihood estimation (MLE) are two popular estimation methods, see \cite{bai2003inferential} for the former and \cite{bai2012statistical} for the latter as two examples. In this paper we focus on exploring the PCA-based method. A conventional PCA solution to model (\ref{equ2.1}) with $p\ge T$ in \cite{bai2003inferential} is to minimize
\[
Q_1(\Bb,\Fb)=(pT)^{-1}\|\Xb-\Fb\Bb^\top\|_F^2,\quad\text{s.t.}\quad \Fb^\top\Fb/T=\Ib_r.
\]
It's well known the solution is $\hat\Bb=T^{-1}\Xb^\top\hat\Fb$ while $\hat\Fb=\arg\min_{\Fb}\text{tr}(\Fb^\top\Xb\Xb^\top\Fb)$. Usually the columns of  $\hat\Fb$ are taken as $\sqrt{T}$ times the leading $r$ eigenvectors of $\Xb\Xb^\top$, while $r$ is predetermined by information criterion in \cite{bai2002determining} or other criteria. Further,  plug-in estimators of the common components are naturally $\hat c_{tj}=\sum_{k=1}^{r}\hat b_{jk}\hat f_{tk}$. Consistency and limiting distributions of $\hat b_{jk}$, $\hat f_{tk}$ and $\hat c_{tj}$ from such a conventional  PCA procedure are obtained in \cite{bai2003inferential} for the large-dimensional setting, with the assumption that  serial and cross-sectional correlations between  the idiosyncratic errors are weak.

\subsection{Factor model with prior network information}
For model (\ref{equ2.1}), convergence rates of conventional PCA estimators are well known to be determined by both $T$ and $p$. The rates in \cite{bai2003inferential} are $T^{-1}\|\hat \Fb-\Fb\Hb\|_F^2=O_p(T^{-2}+p^{-1})$ and $p^{-1}\|\hat\Bb-\Bb\Hb\|_F^2=O_p(T^{-1})$ for $p\gtrsim T$, where $\Hb$ is rotational matrix which depends on $T,p$. Hence it benefits from the so-called \emph{blessing of dimensionality} that a large $p$ improves the estimation.  However,  the generic PCA can lead to unreliable estimators for small $T$.  A scenario with large $p$ but small $T$ is common for the analysis of real data, especially in finance or macroeconomics. For instance, usually weekly but not daily stock returns are studied in order to control serial dependencies, which decreases the size of $T$. On the other hand, it's more likely that the factor system undergoes structural breaks if the sampling period is too long.

Fortunately, reliable estimation is hopeful for small $T$ if priori information is taken into account for the factor model. Incorporating a network structure to show the interconnectivity of massive variables is rather natural and reasonable in real application. Hence, we aim to improve the estimation of factor loadings and factor scores for model (\ref{equ2.1}) when a priori network is given. In this paper, the variable-to-variable network structure is assumed to be observed  according to common knowledge or intrinsic relationship of the variables, rather than estimated from the data matrix $\Xb$. For instance, when we study the stock returns of $p$ companies, an edge between companies $(i,j)$ either represents they are in the same industrial categories or they are stakeholders to each other. We use a $p\times p$ adjacency matrix $\Ab$ to denote the observed network. We set $\mathrm A_{ij}=\mathrm A_{ji}=1$ if there exists an edge between two variables $i$ and $j$ and set $\mathrm A_{ij}=0$ otherwise. In the factor model (\ref{equ2.1}), it's very natural to assume  $\|\bb_i-\bb_j\|^2$  being  small if $\mathrm A_{ij}=1$, where $\bb_i^\top=(b_{i1},\ldots,b_{ir})$ is the loading vector for variable $i$. In other word,  $\mathrm A_{ij}$ reflects the similarity  between $\bb_i$ and $\bb_j$, which indicates how the common factors affect corresponding variables.

In addition to the observed data matrix $\Xb$, now we also have prior network information collected in the adjacency matrix $\Ab$, thus we refer the new factor model as the Network-linked Factor Model (NFM). In next section, we  propose an adaptive procedure to estimate the loadings and scores of the proposed NFM. Intuitively if the network perfectly or partially describes the connection of corresponding variables, we can reduce the uncertainty of loadings and improve the accuracy. Our theoretical analysis will show that not only loadings but also factor scores can be more precisely specified.  On the other hand, even if the network is totally misleading, our approach still performs comparably to the PCA solutions in \cite{bai2003inferential}, by adaptively controlling the influence of the network information on the estimation procedure.

\section{Adaptive estimation for NFM}\label{sec3}
\subsection{Laplacian penalty and Projection penalty}
We study the estimation of loadings and factor scores under the NFM framework in this section. If $r$ is given, a natural idea is to penalize the loadings in the objective function $Q_1(\Bb,\Fb)$, which leads to
\[
Q_2(\Bb,\Fb)=\frac{1}{pT}\|\Xb-\Fb\Bb^\top\|_F^2+\tilde\alpha \sum_{i}\sum_j\mathrm A_{ij}\|\bb_i-\bb_j\|^2,\quad\text{s.t.}\quad\Fb^\top\Fb/T=\Ib_r,
\]
where $\tilde\alpha$ is a tuning parameter taken as given for now. To better explain how the penalty regularizes the estimation, the  Laplacian matrix of $\Ab$ is defined as $\mathcal{L}=\mathcal{D}-\Ab$, where $\mathcal{D}$ is a diagonal matrix with entries $\mathcal{D}_{ii}=\sum_{j=1}^p\mathrm{A}_{ij}$. Meanwhile, define the normalized Laplacian as $\mathcal{L}_n=\bar d^{-1}\mathcal{L}$ with $\bar d=p^{-1}\sum_{i=1}^{p}\mathcal{D}_{ii}$. It's easy to verify that for the symmetric adjacency matrix $\Ab$,
\begin{equation}\label{equ3.1}
\sum_{i=1}^{p}\sum_{j=1}^p\mathrm{A}_{ij}\|\bb_i-\bb_j\|^2=2\text{tr}(\Bb^\top\mathcal{L}\Bb)=2\bar d\text{tr}(\Bb^\top\mathcal{L}_n\Bb).
\end{equation}
Suppose $\mathcal{L}_n$ has the spectral decomposition as $\mathcal{L}_n=\Ub\bLambda\Ub^\prime$, where the columns of $\Ub=(\bu_1,\ldots,\bu_p)$ are composed of eigenvectors and $\bLambda=\text{diag}(\tau_1,\ldots,\tau_p)$ contains the eigenvalues in decreasing order. Define $\tilde\Xb=\Xb\Ub$, $\tilde\Bb=\Ub^\top\Bb=(\tilde b_{jk})_{p\times r}$ and rescale $\tilde\alpha$ by $\alpha=2p\bar d\tilde\alpha$, then we have
\begin{equation}\label{equ3.2}
Q_2(\Bb,\Fb)=\frac{1}{pT}\|\tilde\Xb-\Fb\tilde\Bb^\top\|_F^2+\frac{\alpha}{p} \sum_{k=1}^r\sum_{j=1}^{p}\tau_j\tilde b_{jk}^2.
\end{equation}
We refer the penalty term in $Q_2(\Bb,\Fb)$ to the Laplacian penalty.

It's clear that in equation (\ref{equ3.2}) the penalization to $\tilde b_{jk}$ depends on corresponding eigenvalues $\tau_j$, where $\tilde b_{jk}$ are representations of $b_{jk}$ under the new basis $\Ub$. For larger $\tau_j$, the penalization becomes stronger. In another word, we actually  expect $|\tilde b_{jk}|$ to be small if $\tau_j$ is large. By the decomposition $\Bb=\sum_{j=1}^{p}\bu_j\tilde\bb_j^\top$, the Laplacian penalty regularizes the estimated loading matrix to lie closer to the space spanned by $\bu_j$ corresponding to smaller eigenvalues, especially when these small eigenvalues are sufficiently close to 0 and well distinguished from those spiked ones.

Though the Laplacian penalty is frequently adopted for network analysis, it may not always be the best for our NFM framework. The penalty subjectively supposes $\tilde b_{jk}$ are related to $\tau_j$, and assigns different weights based on the eigenvalues. However, in real cases the given network can be  {misspecified} so that  $\tilde b_{jk}$ are not necessarily determined by $\tau_j$. A more general penalty should put separate tuning parameters to $\tilde b_{jk}$, such as $\sum_{k=1}^{r}\sum_{j=1}^{p}p^{-1}\alpha_j\tilde b_{jk}^2$, but it contains $p$ tuning parameters and is impractical in real application. In fact, the Laplacian penalty can be viewed as a special and feasible case by taking  $\alpha_j=\alpha \tau_j$.

 We propose another alternative penalty, named the Projection penalty. For the Projection penalty, the eigenvalues are divided into two groups $\{\tau_1,\ldots,\tau_{p-m}\}$ and $\{\tau_{p-m+1},\ldots,\tau_p\}$ by a preliminarily determined \emph{truncation parameter} $m$. Then, we penalize $Q_1(\Bb,\Fb)$ by $p^{-1}\alpha \sum_{k=1}^{r}\sum_{j=1}^{p-m}\tilde b_{jk}^2$. Compared with the Laplacian penalty, now we assign equal penalization to $\tilde b_{jk}$ if $\tau_j$ is large, but no penalization to the small part. In this way,  the estimated  $\Bb$ are still more likely to lie in the space spanned by $\bu_j$ corresponding to small eigenvalues, or more precisely by $(\bu_{p-m+1},\ldots,\bu_p)$. To see why we name it as Projection penalty, separate $\Ub$ as $\Ub=(\Ub_1,\Ub_2)$ where $\Ub_1$ is composed of the leading $(p-m)$ eigenvectors. Then, the penalty can be written in matrix form as
\[
\frac{\alpha}{p} \sum_{k=1}^{r}\sum_{j=1}^{p-m}\tilde b_{jk}^2=\frac{\alpha}{p}\|\Ub_1(\Ub_1^\top\Ub_1)^{-1}\Ub_1^\top\Bb\|^2,
\]
where $\Ub_1(\Ub_1^\top\Ub_1)^{-1}\Ub_1^\top$ is exactly the projection matrix to the column space of $\Ub_1$.

\subsection{Comparison between Laplacian and Projection penalties}
Both Laplacian and Projection penalties are motivated from the priori network information, but they show different understandings of the loading structure. Neither of them can always outperform the other, as the accurate relationship between the network and factor loadings remains unknown. However, both the penalties result in more efficient estimators compared with  conventional PCA  solutions for certain cases. We compare the penalties in the following three aspects, while more sophisticated theoretical results are presented in Section \ref{sec4}.

Firstly, both of them are special cases of the more general penalty $p^{-1}\sum_{k=1}^{r}\sum_{j=1}^{p}\alpha_j\tilde b_{jk}^2$. The Laplacian penalty takes $\alpha_j=\alpha\tau_j$ while the Projection penalty takes $\alpha_j=\alpha$ for $j\le p-m$ and $\alpha_j=0$ for $j>p-m$. Hence if $\tilde b_{jk}$ are truly small when $\tau_j$ are large, the Laplacian penalty can be more efficient due to the proper weights. Otherwise, the network information is abused by Laplacian penalty with undesirable misspecification.

Secondly,  both penalties tend to project the loading matrix to the space spanned by the last several eigenvectors of $\mathcal{L}_n$.  However, the regularization of Laplacian penalty seems to be stronger as it would force $\hat b_{ik}=\hat b_{jk}$ for all $i,j\le p$ if the network is  connected and taking $\alpha\rightarrow\infty$. This would result in an estimated  loading matrix of rank one, i.e., there exists only one factor in the system. In contrast,  the Projection penalty only forces that the estimated $\hat\Bb$  lies in the space of $\Ub_2$ under the same scenario, which is of rank $m$ and thus more flexible.

Thirdly,  both the penalties involve tuning parameters to be selected and are  computationally more demanding than conventional PCA. In addition, the Projection penalty has one more tuning parameter than the Laplacian penalty.

The relationship between the two penalties are analogous to  that between the ridge regression and principal component regression (PCR). Ridge regression  assigns different weights to the singular vectors of the design matrix, while the weights  depend on the corresponding singular values, similar to  the Laplacian penalty here. On the other hand, PCR only picks up the leading singular vectors and assigns equal importance to them,  similar to the Projection penalty here.  Ridge regression and PCR  both work in existence of multi-collinearity, while the  Laplacian and Projection penalties here are mainly imposed to incorporate prior loading similarity structure in the factor model.

\subsection{Solutions to the optimization problems}\label{sec3.3}
We first show the solution with Laplacian penalty, that is, to minimize $Q_2(\Bb,\Fb)$ in equation (\ref{equ3.1}). First assume $\alpha$ and $\hat\Fb_{lap}$ are given while $T^{-1}\hat\Fb_{lap}^\top\hat\Fb_{lap}=\Ib_r$, then
\begin{equation}\label{equ3.3}
\begin{split} Q_2(\Bb,\hat\Fb_{lap})=&\frac{1}{pT}\text{tr}\Big(\Xb^\top\Xb-\Xb^\top\hat\Fb_{lap}\Bb^\top-\Bb\hat\Fb_{lap}^\top\Xb+\Bb\hat\Fb_{lap}^\top\hat\Fb_{lap}\Bb^\top\Big)+\frac{\alpha}{p}\text{tr}(\Bb^\top\mathcal{L}_n\Bb)\\
=&\frac{1}{pT}\text{tr}(\Xb^\top\Xb)-\frac{1}{pT^2}\text{tr}(\Xb^\top\hat\Fb_{lap}\hat\Fb_{lap}^\top\Xb)+\frac{1}{p}\|T^{-1}\Xb^\top\hat\Fb_{lap}-\Bb\|_F^2+\frac{\alpha}{p}\text{tr}(\Bb^\top\mathcal{L}_n\Bb).
\end{split}
\end{equation}
Therefore, the solution for loadings should be
\[
\hat\Bb_{lap}=\frac{1}{T}(\Ib_p+\alpha\mathcal{L}_n)^{-1}\Xb^\top\hat\Fb_{lap}.
\]
Take $\Bb=T^{-1}(\Ib_p+\alpha\mathcal{L}_n)^{-1}\Xb^\top\Fb$ back to $Q_2(\Bb,\Fb)$ so that
\[ Q_2(\Bb,\Fb)=\frac{1}{pT}\text{tr}(\Xb^\top\Xb)-\frac{1}{pT^2}\text{tr}\Big\{\Fb^\top\Xb(\Ib_p+\alpha\mathcal{L}_n)^{-1}\Xb^\top\Fb\Big\}.
\]
Thus, the solution for $\hat\Fb_{lap}$ can be taken as $\sqrt{T}$ times the leading $r$ eigenvectors of $\Xb(\Ib_p+\alpha\mathcal{L}_n)^{-1}\Xb^\top$.  For the Projection penalty, simply replace $\text{tr}(\Bb^\top\mathcal{L}_n\Bb)$ in (\ref{equ3.3}) with $\text{tr}(\Bb^\top\Ub_1\Ub_1^\top\Bb)$ and  similar procedure as above can be taken to obtain the solution. Denote the solution as $\hat\Fb_{proj}$ and $\hat\Bb_{proj}$, where $\hat\Fb_{proj}$ is $\sqrt{T}$ times the leading $r$ eigenvectors of $\Xb(\Ib_p+\alpha\Ub_1\Ub_1^\top)^{-1}\Xb^\top$ and $\hat\Bb_{proj}=T^{-1}(\Ib_p+\alpha\Ub_1\Ub_1^\top)^{-1}\Xb^\top\hat\Fb_{proj}$. The common components are then given directly by $\hat c_{ij}=\sum_{k=1}^{r}\hat b_{ik}\hat f_{tk}$. To simplify the notations for the rest of this paper, we define $\Db_1(\alpha)=\Ib_p+\alpha\mathcal{L}_n$, $\Db_2(\alpha,m)=\Ib_p+\alpha\Ub_1\Ub_1^\top$, and use $\Db_1$ (or $\Db_2$) as long as it doesn't lead to misunderstanding.

As a consequence, the closed-form solutions exist for the penalized method, {which benefits from the $\ell_2$-type penalty}. As can be seen from the solutions, both  the penalized methods have ``shrinkage" effect on the estimators. Shrinkage estimators are usually biased but may have lower mean squared error by the trade-off between bias and variance. Similar phenomenon appears here with carefully chosen tuning parameters, so that the common factors will remain powerful while the idiosyncratic errors are eliminated.  {Regularized method for factor model is also proposed in  \cite{bai2019rank}, but we clarify this paper differentiates from \cite{bai2019rank} with respect to the research purpose, the interpretation of penalties and the estimation procedure.}

The closed-form solutions also simplify the computation and parameter selection for the proposed penalized approaches. In detail, we can keep the matrices  $\Ub$, $\bLambda$ and $\tilde \Xb=\Xb\Ub$ to avoid repetitively  computing inverses of $p\times p$ matrices when selecting the tuning parameters.  For example, {$\hat\Fb_{lap}$ and $\hat\Bb_{lap}$ can be more easily computed by   $\Xb\Db_1^{-1}\Xb^\top=\tilde\Xb(\Ib_p+\alpha\bLambda)^{-1}\tilde\Xb^\top$ and $\hat\Bb_{lap}=T^{-1}\Ub(\Ib_p+\alpha\bLambda)^{-1}\tilde\Xb^\top\hat\Fb_{lap}$ as $(\Ib_p+\alpha\bLambda_L)$ is diagonal. }Similar technique can be applied to the optimization problem with Projection penalty, by noting that $\Db_2^{-1}=\Ub\{\Ib_p+\alpha\text{diag}(\Ib_{p-m},{\bf 0})\}^{-1}\Ub^\top$.

\section{Theoretical results}\label{sec4}
This section aims to explore the asymptotic properties of the shrinkage estimators corresponding to  the Laplacian and Projection penalties respectively. We first present some mild conditions, which bring into a large-dimensional framework with both serially and cross-sectionally correlated errors. Then, consistency rates of the new estimators are studied, followed by comparisons to conventional PCA solution. Furthermore, when certain conditions on central limit theorems are satisfied, we prove the asymptotic normalities of the resulting estimators.  It will be concluded that the estimation of factor scores, loadings and common components can all be improved when the network structure truly reflects the relationships among variables. Meanwhile, the new approaches are indeed adaptive in the sense that even if the network information is incorrect/misleading, they just result in almost the same efficient estimators as conventional PCA solution.

\subsection{Assumptions}
Diversified assumptions have been proposed in the literature so as to cope with different scenarios and motivations on factor models.  For instance, \cite{bai2002determining} and \cite{bai2003inferential} directly assumed bounded moment constraints on factors and idiosyncratic errors, while \cite{fan2013large} focused more on the data generating process and tail probability bounds. In this paper, we adopt the following general and interpretable assumptions.
\begin{description}
	\item[Assumption A] We assume that $p\gtrsim T$, and $r$ is assumed to be constant {while $p,T\rightarrow\infty$}.
\end{description}
\begin{description}
	\item[Assumption B] For all $t\le T$ and $k\le r$, ${\rm E}(f_{tk})=0$, ${\rm E}(f_{tk}^2)=1$ and ${\rm E}(f_{tk}^4)\le M$ for some constant $M>0$. Further assume $\Fb^\top\Fb/T=\Ib_r$ almost surely and for any $T$ dimensional vector $\bv$ that $\|\bv\|=1$, we have ${\rm E}\|\bv^\top\Fb\|^2\le M$.
\end{description}
\begin{description}
	\item[Assumption C] For any $j\le p$ and $k\le r$, assume $|b_{jk}|\le M$. Further assume $p^{-1}\Bb^\top\Bb\rightarrow \bSigma_B$, while the eigenvalues of $\bSigma_B$ are distinct and bounded away from zero and infinity, i.e., $c_1\ge\lambda_1(\bSigma_B)>\cdots>\lambda_r(\bSigma_B)\ge c_2$, for some constants $c_1,c_2>0$.
\end{description}

\begin{description}
	\item[Assumption D]  The error matrix $\Eb=\Pb_1{\bf\mathcal{E}}\Pb_2$, where $\mathcal{E}=(\epsilon_{tj})_{T\times p}$ with $\epsilon_{tj}$ being independent  and ${\rm E}(\epsilon_{tj})=0$, ${\rm E}(\epsilon_{tj}^2)=1$, ${\rm E}(\epsilon_{tj}^4)\le M$, $\Pb_1$ and $\Pb_2$ are two deterministic square matrices. There exists positive constants $c_1$ and $c_2$ so that $c_2\le\lambda_t(\Pb_1^\top\Pb_1)\le c_1$ for any $t\le T$ and $c_2\le\lambda_j(\Pb_2^\top\Pb_2)\le c_1$ for any $j\le p$. In addition, ${\bf \mathcal{E}}$ and $\Fb$ are independent.
\end{description}	

Due to the existence of  prior network information, it's more meaningful to consider the large $p$ and possibly small $T$ scenarios. We clarify this in Assumption A, where the number of factors shall be predetermined and taken as fixed similar to  most existing works. The practical selection of factor number $r$ is discussed in the next section. Assumption B requires the latent factors have bounded fourth moments, which can easily hold, for example when the factors are from multiple time series models such as the VAR process. By assuming that $\Fb^\top\Fb=\Ib_r/T$ almost surely and ${\rm E}\|\bv^\top\Fb\|^2\le M$ for any $\|\bv\|=1$, we essentially require that the serial dependence among factors can not be too strong, and guarantee the model is identifiable up to orthogonal transformations.  Assumption C is quite common in the factor model literature. The condition of distinct and bounded eigenvalues ensures the factors are sufficiently powerful and the corresponding eigenvectors are identifiable. Assumption D allows both serial and cross-sectional dependences for the idiosyncratic errors, even though  the entries of $\mathcal{E}$ are independent. Similar assumption can be found in \cite{bai2006determining} and \cite{han2017determining}.

We have obtained the solutions corresponding to the  Laplacian and Projection penalties in section \ref{sec3}. Compared with conventional PCA solution which is based on $\Xb\Xb^\top$, the new solutions are established on a shrunk version $\Xb\Db^{-1}\Xb^\top$, where $\Db=\Db_1$ or $\Db_2$. Basically the factors should remain informative after shrinkage. To this end, we introduce the following assumption.
\begin{description}
	\item[Assumption E] Define ${\bf S}=p^{-1}\Bb^\top\Db^{-1}\Bb$ where $\Db=\Db_1$ (or $\Db_2$), then we assume there exists constant c such that $\lambda_1({\bf S})>\cdots>\lambda_r({\bf S})\ge c$ as $p,T\rightarrow \infty$.
\end{description}
It's easy to verify the major term in $(pT)^{-1}\Xb\Db^{-1}\Xb^\top$ is $\Fb{\bf S}\Fb^\top/T$, which asymptotically shares the same non zero eigenvalues of ${\bf S}$ backed by Assumption B. Hence, Assumption E ensures the power of factors after the shrinkage. Note that $\Db_1$ (or $\Db_2$) contains the tuning parameters, which are usually chosen as functions of the dimensions $p$ and $T$. Consequently, Assumption E performs more like a checking procedure  about whether the tuning parameters are reasonably selected.  Besides, Assumption E always holds if we set $\alpha=0$ as long as Assumption C holds.

\subsection{Consistency}
In this subsection, we establish the consistency of the proposed adaptive estimators. We first present a general theorem on how the consistent rates depend on the tuning parameters.
\begin{theorem}\label{thm4.1}
	Suppose  Assumptions A-E hold, then there exist a sequence of invertible matrices $\Hb$ (depending on $p,T$ and tuning parameters) such that $\Hb^\top\Hb\stackrel{p}{\rightarrow}\Ib_r$ as $p,T\rightarrow\infty$, and
\[
\begin{split}
\frac{1}{T}\|\hat\Fb-\Fb\Hb\|_F^2=&O_p\bigg(\frac{1}{p}+\frac{\text{tr}(\Db^{-2})}{pT^2}\bigg),\\
\frac{1}{p}\|\hat\Bb-\Bb\Hb \|_F^2=&O_p\bigg(\frac{1}{p}\|(\Db^{-1}-\Ib_p)\Bb\|_F^2+\frac{\text{tr}(\Db^{-2})}{pT}\bigg),\\
\frac{1}{pT}\|\hat\Cb-\Cb\|_F^2=&O_p\bigg(\frac{1}{p}\|(\Db^{-1}-\Ib_p)\Bb\|_F^2+\frac{\text{tr}(\Db^{-2})}{pT}+\frac{1}{p}\bigg),
\end{split}
\]
where $\hat\Bb=T^{-1}\Db^{-1}\Xb^\top\hat\Fb$, $\hat\Fb$ is $\sqrt{T}$ times the leading $r$ eigenvectors of $(pT)^{-1}\Xb\Db^{-1}\Xb^\top$ and $\hat\Cb=\hat\Fb\hat\Bb^\top$. In addition,  $\Db=\Db_1=\Ib_p+\alpha\mathcal{L}_n$ for Laplacian penalty, while $\Db=\Db_2=\Ib_p+\alpha\Ub_1\Ub_1^\top$ for Projection penalty with $\Ub_1=(\bu_1,\ldots,\bu_{p-m})$, where $\bu_j$ is the $j$-th eigenvector of $\mathcal{L}_n$.
\end{theorem}
Theorem \ref{thm4.1} provides unified results for both the Laplacian and Projection penalties. Actually this theorem makes it possible to analyse all penalties with a form of $p^{-1}\alpha\text{tr}(\Bb^\top\Db\Bb)$. For the factor score matrix, compared with the convergence rate $O_p(p^{-1}+T^{-2})$ in \cite{bai2003inferential}, our methods can be more accurate if $\text{tr}(\Db^{-2})=o(p)$ and $T=o(\sqrt{p})$. For the loading matrix, compared with the conventional rate $O_p(T^{-1})$ (we assume $p\gtrsim T$), the new consistency rate is composed of two parts. The first part corresponds to shrinkage bias, which depends on the tuning parameters, and is driven by the relations between the observed network and the latent loading matrix.  We claim that with suitably selected tuning parameters,  the shrinkage bias is always negligible.  The second term is a scaling of $T^{-1}$, while the scaling coefficient depends on the tuning parameters. Note that when $\alpha$ is large, $\text{tr}(\Db^{-2})/p$ can be small and even converge to 0, so the new estimators are impressively more efficient. The consistency rate for common components is obtained by combining the convergence rates of the loadings and factor scores, where the loadings' part always dominates because of  $p\gtrsim T$. We emphasize that by Theorem \ref{thm4.1},  even if the network is totally incorrect, we can always obtain the convergence rates as for the conventional PCA solutions, by setting $\alpha\approx 0$. On the other hand, if the network precisely capture the interconnectivity of the variables, our approach achieves significant improvements. Therefore, the new estimators achieve adaptivity to the efficiency and accuracy of the prior network information.

Theorem \ref{thm4.1} also makes it easier to study the optimal selection of tuning parameters. Given the consistent rates, the tuning parameters can be selected naturally by minimizing $h(\Db)$:
\begin{equation}\label{equ4.1}
h(\Db)=\frac{1}{p}\|(\Db^{-1}-\Ib_p)\Bb\|_F^2+\frac{\text{tr}(\Db^{-2})}{pT}.
\end{equation}
To this end, the Laplacian penalty and Projection penalty should be treated separately. We start with the Projection penalty, so $\Db=\Db_2=\Ib_p+\alpha\Ub_1\Ub_1^\top$, where $\Ub_1$ is composed of the leading $(p-m)$ eigenvectors of the normalized Laplacian matrix. Using the notation $\tilde\Bb_1=\Ub_1^\top\Bb$, we have
\[
h(\Db_2)=\frac{\alpha^2}{p(1+\alpha)^2}\|\tilde\Bb_1\|_F^2+\frac{1}{T(1+\alpha)^2}+O\bigg(\frac{m}{pT}\bigg).
\]
Typically $m$  is set as a small number compared to $p$, thus we ignore the term $m/(pT)$. Hence, $h(\Db_2)$ is approximately minimized by taking $\alpha=p/(T\|\tilde\Bb_1\|_F^2)$  so that
\[
h(\Db_2)=\frac{(pT)^{-1}\|\tilde\Bb_1\|_F^2}{p^{-1}\|\tilde\Bb_1\|_F^2+T^{-1}}+O\bigg(\frac{m}{pT}\bigg).
\]
It's necessary to guarantee that  Assumption E holds with this ``optimal" $\alpha$. The following corollary confirms the above judgement.
\begin{corollary}\label{cor1}
	Given $m$, Assumption E always holds with $\alpha=p/(T\|\tilde\Bb_1\|_F^2)$ for Projection penalty as long as Assumption C holds.
\end{corollary}	
Consequently, the consistency for Projection based estimator depends on the size of $p^{-1}\|\tilde\Bb_1\|_F^2$, which corresponds to the shrinkage bias. When the projection is reasonable and precise, $p^{-1}\|\tilde\Bb_1\|_F^2$ would be small, i.e., $p^{-1}\|\tilde\Bb_1\|_F^2=o(T^{-1})$, and the consistent rate for the estimated loadings and common components can be improved. It should also be  pointed out that even under the worst case $p^{-1}\|\tilde\Bb_1\|_F^2\gg T^{-1}$, we still have $h(\Db_2)\le cT^{-1}$ for some constant $c$, which further indicates the estimators would still be efficient. Actually, the new approach still tends to outperform conventional PCA method even in this case based on our simulations.

As for Laplacian penalty, the discussion is  more complicated. Plug in $\Db_1=\Ib_p+\alpha\bLambda$, then $h(\Db_1)$ has the following form:
\[
h(\Db_1)=\frac{1}{p}\sum_{j=1}^{p}\frac{\alpha^2\tau_j^2\|\tilde\bb_j\|^2}{(1+\alpha\tau_j)^2}+\frac{1}{pT}\sum_{j=1}^{p}\frac{1}{(1+\alpha\tau_j)^2}.
\]
By taking derivative to $\alpha$, we have
\begin{equation}\label{equ4.2}
\frac{\partial h(\Db_1)}{\partial \alpha}=p^{-1}\sum_{j=1}^{p}\frac{2\tau_j(\alpha\tau_j\|\tilde \bb_j\|^2-T^{-1})}{(1+\alpha\tau_j)^3}.
\end{equation}
Hence, there may be multiple local minimum points for $h(\Db_1)$, making it impossible to find the ``optimal" $\alpha$. However, $h(\Db_1)$ is always monotonically decreasing when $\alpha\le\beta$, where $\beta= (T\max_{j}\{\tau_j\|\bb_j\|^2\})^{-1}$. It implies the estimation error can be smaller than conventional PCA solutions when $\alpha\in (0,\beta]$. We consider a sub-optimal tuning parameter $\alpha$ by taking $\alpha=\beta$. Actually $\beta$ is also the exact global minimum point under the special case where $\tau_j\|\tilde\bb_j\|^2$ are  identical for all $\tau_j\ne 0$. By taking $\alpha=\beta$, we then have
\[
h(\Db_1)\le\frac{1}{pT}\sum_{j=1}^{p}\frac{1}{1+\alpha\tau_j}.
\]
Parallelly,
 we also have the following corollary which  ensures that Assumption E holds with this sub-optimal $\alpha$.
\begin{corollary}\label{cor2}
	Assumption E always holds with $\alpha=(T\max_{j}\{\tau_j\|\bb_j\|^2\})^{-1}$ for Laplacian penalty as long as Assumption C holds.
\end{corollary}	
Consequently, the accuracy of Laplacian based estimator depends on the sizes of $\tau_j\|\tilde\bb_j\|^2$. For the case $\max_{j}\tau_j\|\tilde\bb_j\|^2$ is small and sufficiently large number of $\tau_j$ are valid ($\tau_j$ do not vanish with $p\rightarrow \infty$), the penalized method can be  more reliable than conventional PCA method. Note that small $\max_{j}\tau_j\|\tilde\bb_j\|^2$ implies $\text{tr}(\Bb^\top\mathcal{L}_n\Bb)$  is also small. In addition, the adaptivity preserves for Laplacian based estimator such that it  always performs no worse than conventional PCA estimator even if the network is totally incorrect/misleading. However, the lack of theoretically ``optimal" $\alpha$ makes it challenging to directly compare the Laplacian penalty with Projection penalty. We discuss a special case in our supplementary materials, where the asymptotic properties for both methods can be easily studied and compared.

\subsection{Asymptotic normality }\label{sec4.3}
As we have claimed, the adaptive estimators may bring shrinkage biases to the system, so it's more meaningful to investigate  the trade-off between asymptotic biases and variances. In this section, we study the asymptotic normality of the proposed approaches. Additional assumptions on central limit theorem  are necessary for further discussions.
\begin{description}
	\item[Assumption F1] For any $t\le T$,
		\[
				\frac{1}{\sqrt{p}}\Bb^\top\Db^{-1}\be_t\stackrel{d}{\rightarrow}\mathcal{N}({\bf 0},\Vb_t),
		\]
where $\Vb_t=\lim_{p,T\rightarrow\infty}p^{-1}\text{cov}(\Bb^\top\Db^{-1}\be_t)$, $\bp_{1,t}$ is the $t$-th row of $\Pb_1$. Note that $\Db^{-1}$ contains tuning parameters which may depend on $p,T$;
	\end{description}
\begin{description}
	\item[Assumption F2] For any $j\le p$,
		\[
		\frac{1}{\sqrt{T}}\Fb^\top\Eb\frac{\bd_j^{-1}}{\|\bd_j^{-1}\|}=\frac{1}{\sqrt{T}}\sum_{t=1}^{T}\bF_t\be_t^\top\frac{\bd_j^{-1}}{\|\bd_j^{-1}\|}\stackrel{d}{\rightarrow}\mathcal{N}({\bf 0},\Wb_j),
		\]
		where $\bd_j^{-1}$ is the $j$-th column of the matrix $\Db^{-1}$ and $\Wb_j=\lim_{p,T\rightarrow \infty}T^{-1}\text{cov}(\Fb^\top\Eb\bd_j^{-1}/\|\bd_j^{-1}\|)$.
\end{description}

The Assumption F1 is for the asymptotic normality of estimated factor scores, while Assumption F2 is for the loadings.   Validation of the assumptions is out of the scope of current paper, but we claim neither of  them  is stringent while similar assumptions are adopted for obtaining the limiting distributions of PCA estimators  in \cite{bai2003inferential}.  Asymptotic variances in the above assumptions can be further proved to be bounded and non-degenerate in later theorems.
\begin{theorem}\label{thm4.2}
	Denote the spectral decomposition ${\bf S}=\bGamma_S\bLambda_S\bGamma_S^\top$, with ${\bf S}$ defined in Assumption E. When Assumptions A-E and Assumption F1 hold,  for the estimated  factor scores, we have
	\begin{enumerate}
		\item if $\|\Db^{-1}\|_F/T=o(1)$, then
			\[
		\sqrt{p}(\hat\bF_t-\Hb^\top\bF_t)\stackrel{d}{\rightarrow}\mathcal{N}({\bf 0}, \bLambda_S^{-1}\bGamma_S^\top\Vb_t\bGamma_S\bLambda_S^{-1}),
		\]
		where $\Hb$ is the same as in theorem \ref{thm4.1} and $\Vb_t$ is defined in Assumption F1. It can be shown $\Vb_t=\|\bp_{1,t}\|^2\lim_{p,T\rightarrow\infty}p^{-1}(\Bb^\top\Db^{-1}\Pb_2^\top\Pb_2\Db^{-1}\Bb)$, where $\bp_{1,t}$ is the $t$-th row of $\Pb_1$. $\Pb_1$ and $\Pb_2$ are defined in Assumption D.
		\item 	 if $\|\Db^{-1}\|_F/T\ge O(1)$, then $\hat\bF_t-\Hb^\top\bF_t=O_p\big(\|\Db^{-1}\|_F/(T\sqrt{p})\big)$.
	\end{enumerate}
\end{theorem}

The first part of Theorem \ref{thm4.2} shows that the biases of the estimated factor scores are asymptotically negligible  when $\|\Db^{-1}\|_F/T=o(1)$. Note that $\|\Db^{-1}\|_F$ is smaller than $\sqrt{p}$ if the penalty is strong (then $\alpha$ shall be large), so the asymptotic normality can also hold for $p\gtrsim T^2$, while the conventional boundary is $p=cT^2$ in \cite{bai2003inferential}. On the other hand, the variances of the new approaches  are asymptotically equivalent to conventional PCA solution. The second part of Theorem \ref{thm4.2} corresponds to the rate $\text{tr}(\Db^{-2})/(pT^2)$ in Theorem \ref{thm4.1}. Since the consistency has  been discussed in previous subsection, we don't go into details here.

\begin{theorem}\label{thm4.3}
	Suppose Assumptions A-E and Assumption F2 hold and adopt the  notations  in Theorem \ref{thm4.2},   then for the estimated factor loadings,  we have
		\[		\frac{\sqrt{T}(\hat\bb_j-\Hb^\top\Bb^\top\bd_j^{-1})}{\|\bd_j^{-1}\|}\stackrel{d}{\rightarrow}\mathcal{N}({\bf 0},\bGamma_S^\top\Wb_j\bGamma_S),
		\]
		where $\Wb_j$ is defined in Assumption F2 and it can be shown that
		\[
		\Wb_j=\lim_{p,T\rightarrow \infty}\bigg\|\frac{\Pb_2\bd_j^{-1}}{\|\bd_j^{-1}\|}\bigg\|^2\frac{1}{T}{\rm E}(\Fb^\top\Pb_1\Pb_1^\top\Fb).
		\]
\end{theorem}

Theorem \ref{thm4.3} is of importance to explain  why the new methods can perform better.  It's clear that the estimated loadings are actually biased due to the shrinkage, so the trade-off between bias and variance accompanies with the selection of tuning parameters. If we consider finite sample case that $T,p$ are given, for any estimated loading vector $\hat\bb_j$, the bias is $\Hb^\top(\bb_j-\Bb^\top\bd_j^{-1})$ while the variance  is roughly $\|\Pb_2\bd_j^{-1}\|^2T^{-1}{\rm E}(\bGamma_S^\top\Fb^\top\Pb_1\Pb_1^\top\Fb\bGamma_S)$. Hence, we can obtain  the mean squared error (MSE) for the whole loading matrix as
\[
\begin{split}
\text{MSE}\approx&\frac{1}{p}\sum_{j=1}^{p}\bigg\{\|\Hb^\top(\bb_j-\Bb^\top\bd_j^{-1})\|^2+\|\Pb_2\bd_j^{-1}\|^2T^{-1}\text{tr}\bigg({\rm E}(\Fb^\top\Pb_1\Pb_1^\top\Fb)\bigg)\bigg\}
\\=&\frac{1}{p}\|\Hb^\top\Bb^\top(\Ib_p-\Db^{-1})\|_F^2+\frac{1}{pT}\|\Pb_2\Db^{-1}\|_F^2\text{tr}\bigg({\rm E}(\Fb^\top\Pb_1\Pb_1^\top\Fb)\bigg).
\end{split}
\]

By Theorem \ref{thm4.1}, we have $\Hb^\top\Hb\stackrel{p}{\rightarrow}\Ib_r$. Further if $\Pb_1$ and $\Pb_2$ are all identity matrices in Assumption D, which forces an independent and homogeneous case,  we will have
\begin{equation}\label{equ4.3}
\text{MSE}\approx\frac{1}{p}\|\Bb^\top(\Ib_p-\Db^{-1})\|_F^2+\frac{1}{pT}\|\Db^{-1}\|_F^2,
\end{equation}
which is exactly $h(\Db)$ defined in equation (\ref{equ4.1}). Therefore, the ``(sub-)optimal" tuning parameters in last subsection  lead to the minimum MSE for the estimated factor loadings. Actually, the MSE of common components may also be minimized since the errors for estimating loadings seem to dominate. If the idiosyncratic errors are not identically and independently distributed, the ``optimal" tuning parameters still approximately minimize the mean squared error when $p,T$ go to infinity simultaneously. However,  so far the selected tuning parameters depend on unknown loading matrix $\Bb$, which is practically inaccessible. In the next section, we aim to explore data driven methods to determine tuning parameters involved in the estimation procedure.

\section{Tuning parameters}\label{sec5}
We illustrate the criteria for selecting tuning parameters in this section. Basically,  Projection penalty contains three tuning parameters, the \emph{truncation parameter} $m$, the \emph{shrinkage parameter} $\alpha$ and the number of common factors $r$. While for Laplacian penalty, only $\alpha$ and $r$ need to be determined. We  start with $\alpha$ (and $m$), which further provides foundations for  determining  $r$.

\subsection{$C_L$ criterion for $\alpha$}
The shrinkage parameter $\alpha$ is critical to the adaptive estimation and must be carefully chosen. On one hand, if the network is actually irrelevant to the loadings,  the selected $\alpha$ should be nearly 0 in order  to control the shrinkage bias. On the other hand, when the loading matrix is accurately captured by the network structure, more powerful penalization is preferred. When $r$ is given,
 cross validation  procedures can be used to select $\alpha$ (and $m$). We refer to the bandwidth selection in \cite{su2017time} for a typical leave-one-out cross validation (CV) method.  However, the generic CV will suffer from the following two major drawbacks, which motivates us to explore more effective approach. Firstly, the target in CV is usually to minimize $\|\Xb-\hat\Fb\hat\Bb^\top\|_F^2$ at each validation step rather than the real estimation error $\|\Fb\Bb^\top-\hat\Fb\hat\Bb^\top\|_F^2$, so it's contaminated by the idiosyncratic errors. The problem is worthy of more attention since the idiosyncratic errors dominate  in the CV error. Secondly, it's well known the CV procedure is usually computationally inefficient, which makes it less appealing in practice.

 We propose a $C_L$ criterion which avoids the mentioned two disadvantages of generic CV. It borrows the idea of Mallows's $C_L$ (\cite{mallows1973some}), a method  for selecting tuning parameter in ridge regression.  Our criterion is based on the finding that if the idiosyncratic errors are identically and independently distributed ($\Pb_1=\sigma_{e}\Ib_T$, and $\Pb_2=\Ib_p$ in Assumption D), it holds that
\[
\frac{1}{pT}\|\Fb\Bb^\top-\hat\Fb\hat\Bb^\top\|_F^2\approx \frac{1}{pT}\|\Xb-\hat\Fb\hat\Bb^\top\|_F^2-\sigma_e^2+\frac{2r\sigma_e^2}{pT}\text{tr}(\Db^{-1}),
\]
where $\hat\Fb$ and $\hat\Bb$ are the estimated factor score and loading matrix, $\sigma_e^2$ is the variance of $e_{tj}$, while $\Db=\Db_1$ for Laplacian penalty and $\Db=\Db_2$ for Projection penalty  as we  have defined.  Hence, if $\sigma_e^2$ is known, the selected tuning parameters should minimize $\|\Xb-\hat\Fb\hat\Bb\|_F^2+2r\sigma_e^2\text{tr}(\Db^{-1})$. This will directly control the estimation error of common components, and remove all the cross validation process.   If $\sigma_e^2$ is unknown, we can simply replace it with a plug-in estimator, i.e.,
\[
\hat\sigma_e^2=\frac{1}{pT}\|\Xb-\hat\Fb_0\hat\Bb_0^\top\|_F^2,
\]
where $\hat\Fb_0$ and $\hat\Bb_0$ are estimated by conventional PCA method in \cite{bai2003inferential}. In real applications, the idiosyncratic errors may be heterogeneous and serially correlated, but our simulation studies  show the proposed $C_L$ criteria  still works well with weak correlations. Nevertheless, the number of common factors must be determined before applying $C_L$ criteria, which is of separate interest and discussed in the next subsection.

\subsection{``One step further"  approach for specifying $r$}\label{sec:onestepfurther}
Extensive literatures have concerned how to consistently determine the number of common factors. Typically the literatures can be divided into the information criterion-based category (\cite{bai2002determining}, \cite{alessi2010improved}) and eigenvalue-based category(\cite{lam2012factor}, \cite{ahn2013eigenvalue}). One simple way is to directly apply the existing approaches and then take the maximum of the different estimators. However, this may lead to a potential loss of the benefits from priori network.  Therefore, we propose a new approach to determine $r$  by utilizing the network information, which can be more accurate than traditional methods.

The new approach is inspired by the eigenvalue ratio-based estimator (``ER") from \cite{ahn2013eigenvalue}. ``ER" is easily implemented and has been proved promising for the cases with serially and cross-sectionally correlated errors in the literature.  Given a predetermined maximum value $k_{\max}$ for $r$, the ``ER" estimator is simply constructed by
\begin{equation}\label{equ5.1}
\hat r_{er}=\arg\max_{1\le k\le k_{\max}}\frac{\lambda_k(\Xb\Xb^\top)}{\lambda_{k+1}(\Xb\Xb^\top)}.
\end{equation}
Since only the leading $r$ eigenvalues of $\Xb\Xb^\top$ are spiked under the factor structure, the ratio in  equation (\ref{equ5.1}) asymptotically achieves maximization with $k=r$.

 Our method starts with  the estimator  $\hat r_{er}$ in the first step. Then, we use the $C_L$ criterion to select $\alpha$ (and $m$) by setting $r=\hat r_{er}$, and construct two new estimators by
\[
\hat r_{lap}=arg\max_{1\le k\le k_{\max}}\frac{\lambda_k(\Xb\Db_1^{-1}\Xb^\top)}{\lambda_{k+1}(\Xb\Db_1^{-1}\Xb^\top)}, \quad\text{and}\quad\hat r_{proj}=arg\max_{1\le k\le k_{\max}}\frac{\lambda_k(\Xb\Db_{2}^{-1}\Xb^\top)}{\lambda_{k+1}(\Xb\Db_{2}^{-1}\Xb^\top)},
\]
where $\Db_1$ and $\Db_{2}$ are defined as before but  calculated with the selected $\alpha$ (and $m$). We refer this approach as ``one step further".
It's straightforward why it works.
The matrices $\Db_1$ and $\Db_2$ shrink the unspiked eigenvalues while the leading $r$ eigenvalues preserve their power as long as the network is correct. Consequently, the gap between the $r$-th  and $(r+1)$-th eigenvalues is enlarged, which naturally improves the efficiency of eigenvalue ratio-based method in finite sample cases. When  $\hat r_{er}$ in the first step is not accurate, the selected $\alpha$ (and $m$) may not be optimal, but our simulation results demonstrate the  ``one step further" approach still performs convincingly. If that's the case, we may repeat the procedure by starting with $\hat r_{lap}$ or $\hat r_{proj}$ rather than $\hat r_{er}$ until it {becomes sufficiently stable.} From the  simulation studies,  we find the estimator is sufficiently accurate with just one step further.

\section{Simulation studies}\label{sec6}
\subsection{Simulation settings}
The adaptive methods in this paper primarily work for the large $p$-small $T$ cases, so we set $r=3$, $T\in\{20,50\}$, while $p\in\{100,150\ldots,300\}$. Totally 4 cases are designed to comprehensively study the numerical performances  of the proposed approaches with multiple network-loading structures. \\
{\bf Case 1: network structure and factor loadings are generated independently}

The lower-triangular of the adjacency matrix $\Ab$ are \emph{i.i.d.} from $Bernoulli(0.5)$.  The loading matrix is generated by  $B_{jk}\stackrel{i.i.d.}{\sim}\mathcal{N}(0,1)$.\\
{\bf Case 2: factor loadings are generated according to the network}

The adjacency matrix $\Ab$ is generated by the same way as in Case 1. Do the spectral decomposition $\mathcal{L}_n=\Ub\bLambda\Ub^\top$. Define $\Zb_1=(\bu_1,\ldots,\bu_{p-50})$ and $\Zb_2=(\bu_{p-49},\ldots,\bu_p)$. Generate a $(p-50)\times r$ matrix $\bGamma_1$ and $50\times r$ matrix $\bGamma_2$  so that $\bGamma_1^\top\bGamma_1=\bGamma_2^\top\bGamma_2=\Ib_r$. Set $\Bb=0.25\sqrt{s}\Zb_1\bGamma_1+\sqrt{p}\Zb_2\bGamma_2,s=p$.\\
{\bf Case 3: grouped variables}

The $p$ variables are randomly divided into 50 groups with equal probabilities. Let $\mathrm A_{ij}=1$ if and only if  $i$ and $j$ are in the same group. Define $d=\#\{\tau_j<0.001,j\le p\}$, $\Zb_1=(\bu_1,\ldots,\bu_{p-d})$, $\Zb_2=(\bu_{p-d+1},\ldots,\bu_p)$. Generate a $(p-d)\times r$ matrix $\bGamma_1$ by setting $\Gamma_{1,jk}=\tau_j^{-1/2}$ and a $d\times r$ column-orthogonal matrix $\bGamma_2$. Define $s=pr/\|\bGamma_1\|_F^2$ and set  $\Bb=0.25\sqrt{s}\Zb_1\bGamma_1+\sqrt{p}\Zb_2\bGamma_2$.\\
{\bf Case 4: ``active'' and ``inactive" variables}

The last 50 variables are labelled as ``isolated", while the left $(p-50)$ are randomly labelled as ``active" and ``inactive" with equal probabilities. For indexes $i,j$ which are both inactive, generate $\mathrm{A}_{ij}$ from $ Bernoulli(0.1)$, while for those $i,j$ which are both active, set $\mathrm{A}_{ij}=1$ and $\mathrm{A}_{ij}=0$ otherwise. Given the generated adjacency matrix $\Ab$, generate $\Bb$ by the same way as Case 3.

In Case 1, the network and the factor loading matrix are generated separately, thus a totally irrelevant/misleading network is utilized to regularize the estimation. An adaptive method should be able to handle this scenario such that the network information is not abused. In Case 2, the loading matrix is deliberately designed to nearly lie in the space of $\Zb_2$, which is in favour  of the Projection penalty.  Case 3 is motivated by real applications, where the variables are often structured into groups, and we assume two nodes are connected if and only if they are in the same group. Case 4 corresponds to another real case, where some variables are more likely to interact while some are less active or even isolated. The matrices $\bGamma_1$ and $\bGamma_2$ correspond to the projected loadings $\tilde \Bb_1$ and $\tilde\Bb_2$, while the ratio of $\|\tilde\Bb_1\|_F^2$ and  $\|\tilde\Bb_2\|_F^2$ is controlled by $s$. In Case 3 and Case 4, we also deliberately  set the term $\tau_j\|\tilde\bb_j\|^2$ to be identical for $j\le d$. It should be noticed that the eigenvectors of $\mathcal{L}_n$ may  not be unique as many eigenvalues are identical. In the simulation study, we simply use the results calculated by \textsf{R} function \texttt{eigen}. Totally there are 3 network structures incorporated in the 4 cases,  and we illustrate their adjacency matrices  in Figure \ref{fig1} when $p=100$.
\begin{figure}[htbp]
	\centering
	\begin{subfigure}{.3\textwidth}
	\centering
	\includegraphics[width=4cm,height=4cm]{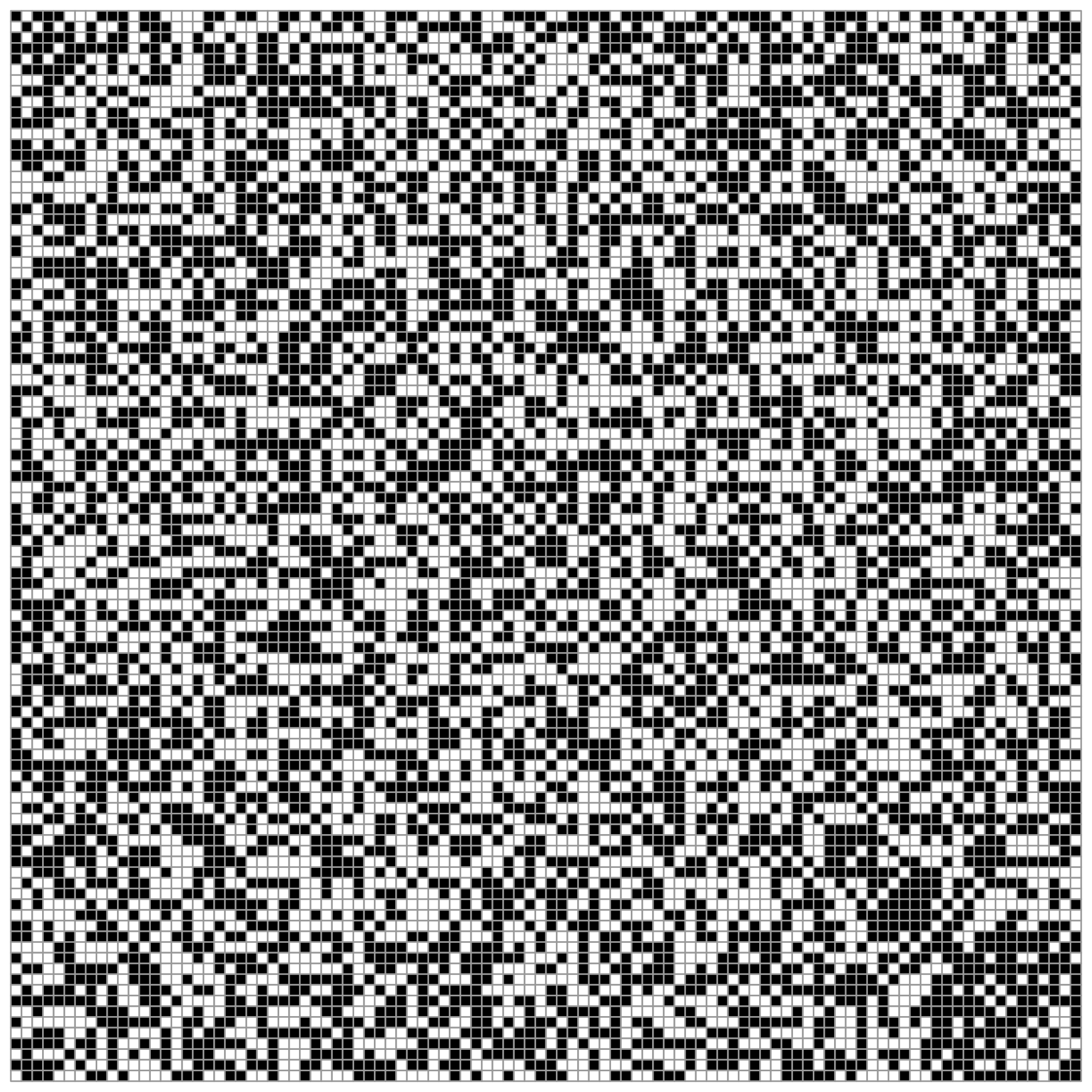}
	\subcaption{Cases 1 and 2}
\end{subfigure}
	\begin{subfigure}{.3\textwidth}
	\centering
	\includegraphics[width=4cm,height=4cm]{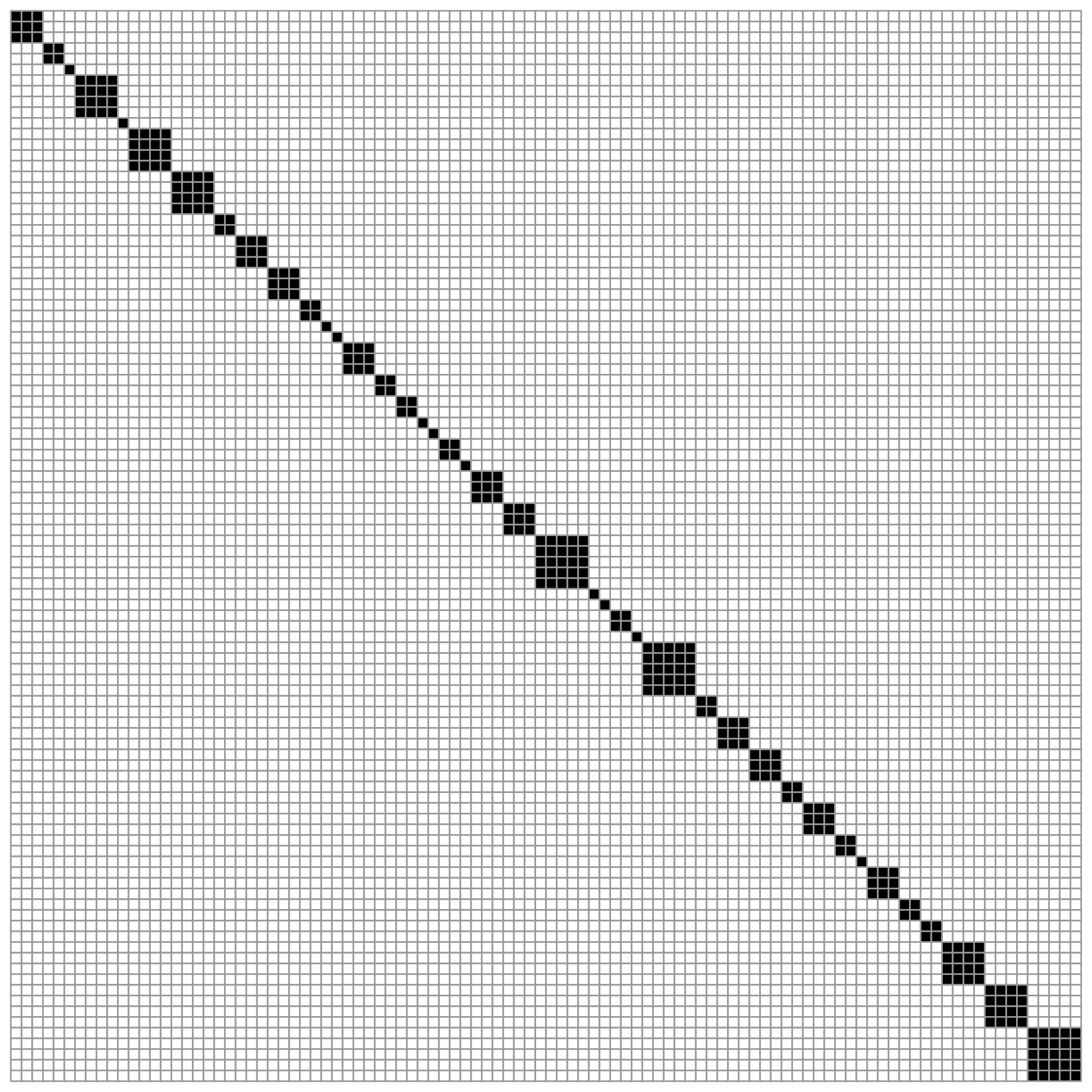}
	\subcaption{Case 3}
\end{subfigure}
	\begin{subfigure}{.3\textwidth}
		\centering
		\includegraphics[width=4cm,height=4cm]{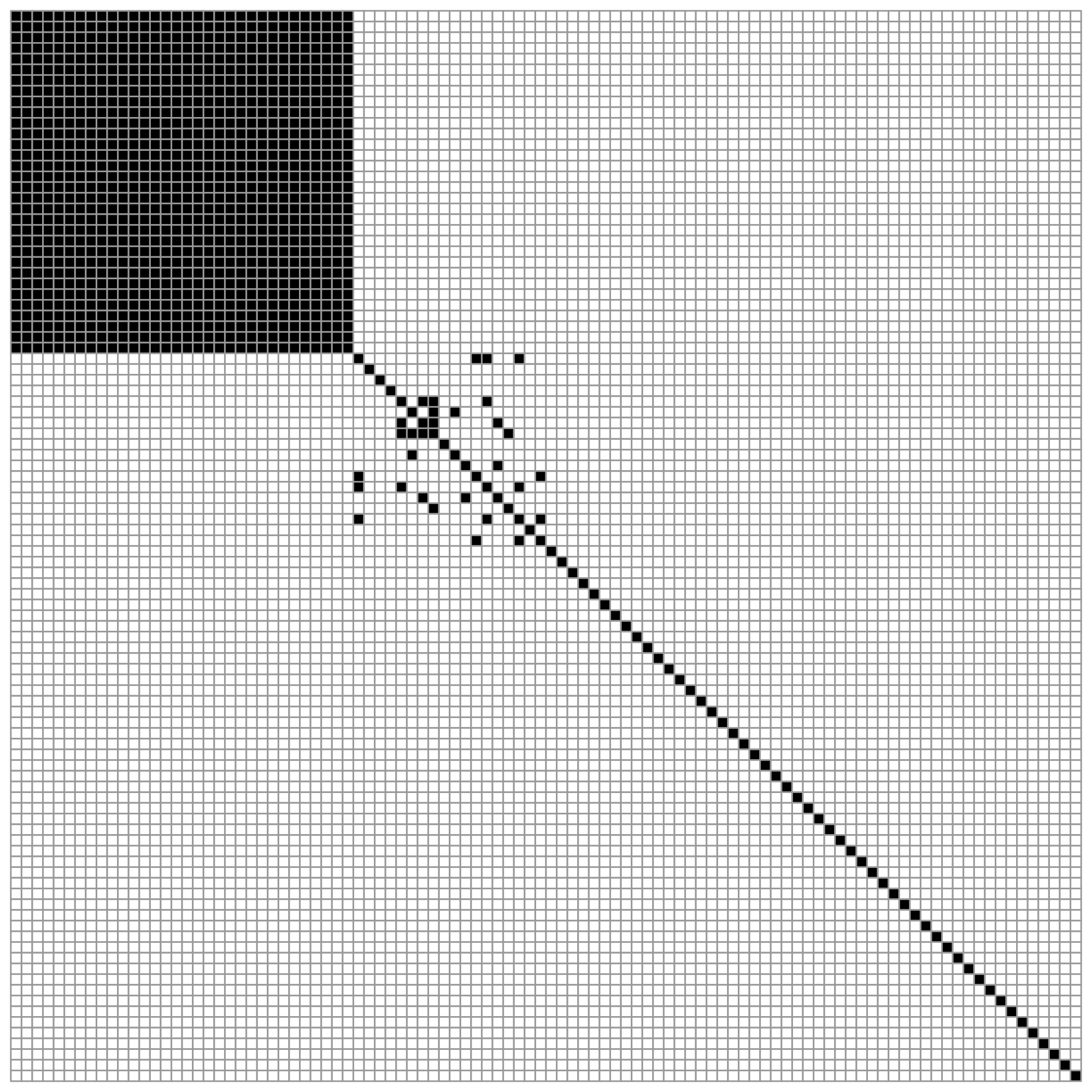}
		\subcaption{Case 4}
	\end{subfigure}
	\caption{The adjacency matrices of the 3 network structures considered in the simulation studies.}
	\label{fig1}
\end{figure}

The factor scores and idiosyncratic errors are always generated based on the following  process. For any $k\le r$, the corresponding factor scores are generated by AR(1) process, specifically by $f_{tk}=0.2f_{t-1,k}+v_{tk}$, where $v_{tk}\stackrel{i.i.d.}{\sim}\mathcal{N}(0,1)$. The idiosyncratic error matrix is produced by $\Eb=\Pb_1\mathcal{E}\Pb_2$, where $\mathcal{E}_{tj}\stackrel{i.i.d.}{\sim}\mathcal{N}(0,\sigma_e^2)$, $\Pb_1$ and $\Pb_2$ are both banded matrices with bandwidth equal to 2 and all the non-zero-off-diagonal entries equal to 0.2. When applying $C_L$ criterion, the parameters $\alpha$ and $m$ are selected from grid-search with $\alpha\in\{p,\frac{1}{b}-1\}$ for $b\in\{0.05,0.10,\ldots,1\}$ and $m\in\{p^{0.1},\ldots,p^{0.9}\}$.

\subsection{Estimation accuracy for common components}\label{sec6.3}
In this section we set $\sigma_e^2=1$ and investigate the estimation accuracy for the common components by our adaptive Laplacian and Projection penalized approaches. The PCA method in \cite{bai2003inferential} is taken as a benchmark. We evaluate their performances by empirical mean squared error (MSE),  which is $(pT)^{-1}\|\hat\Fb\hat\Bb^\top-\Fb\Bb^\top\|_F^2$. All the simulation settings are repeated 500 times and the average MSE is reported in Figure \ref{fig2} when $T=50$. The $T=20$ case leads to similar findings. More detailed results are presented in Section C of the supplementary materials.
\begin{figure}[htbp]
	\begin{subfigure}{.24\textwidth}
	\centering
	\includegraphics[width=4.3cm,height=4.3cm]{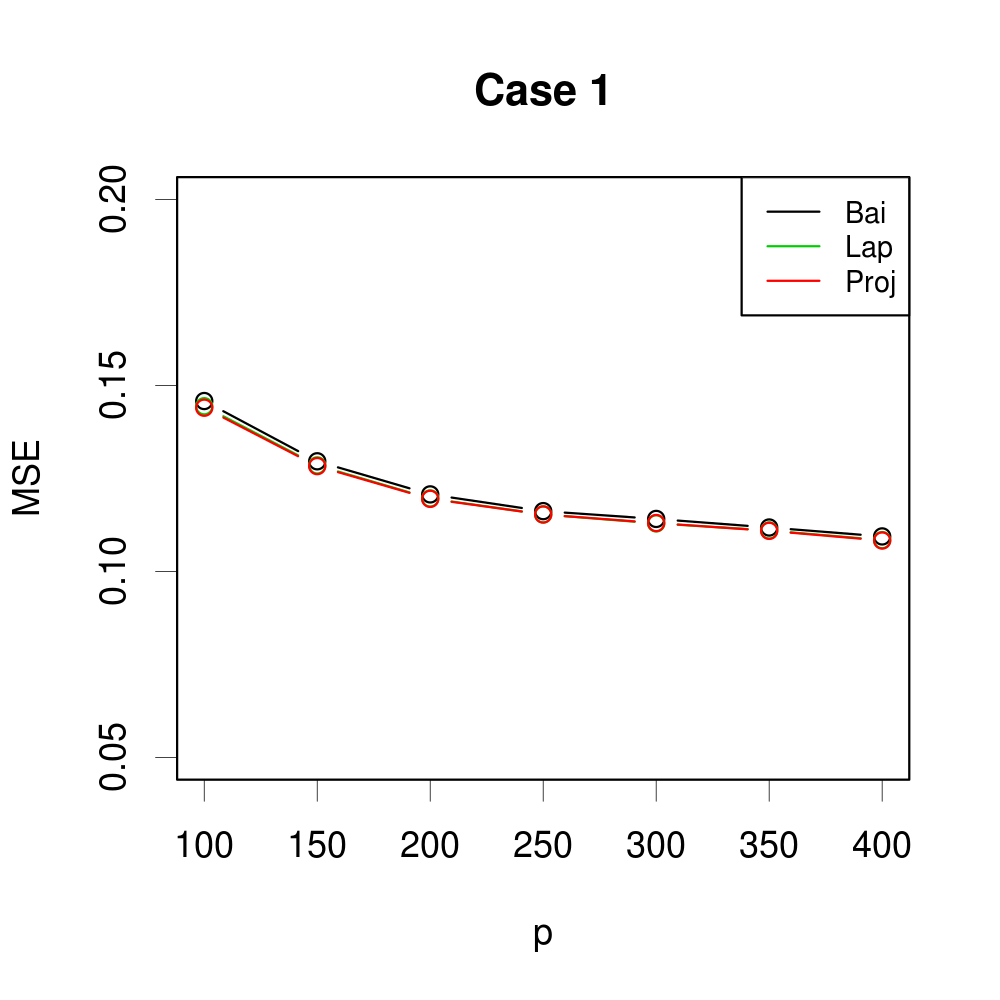}
\end{subfigure}
	\begin{subfigure}{.24\textwidth}
	\centering
	\includegraphics[width=4.3cm,height=4.3cm]{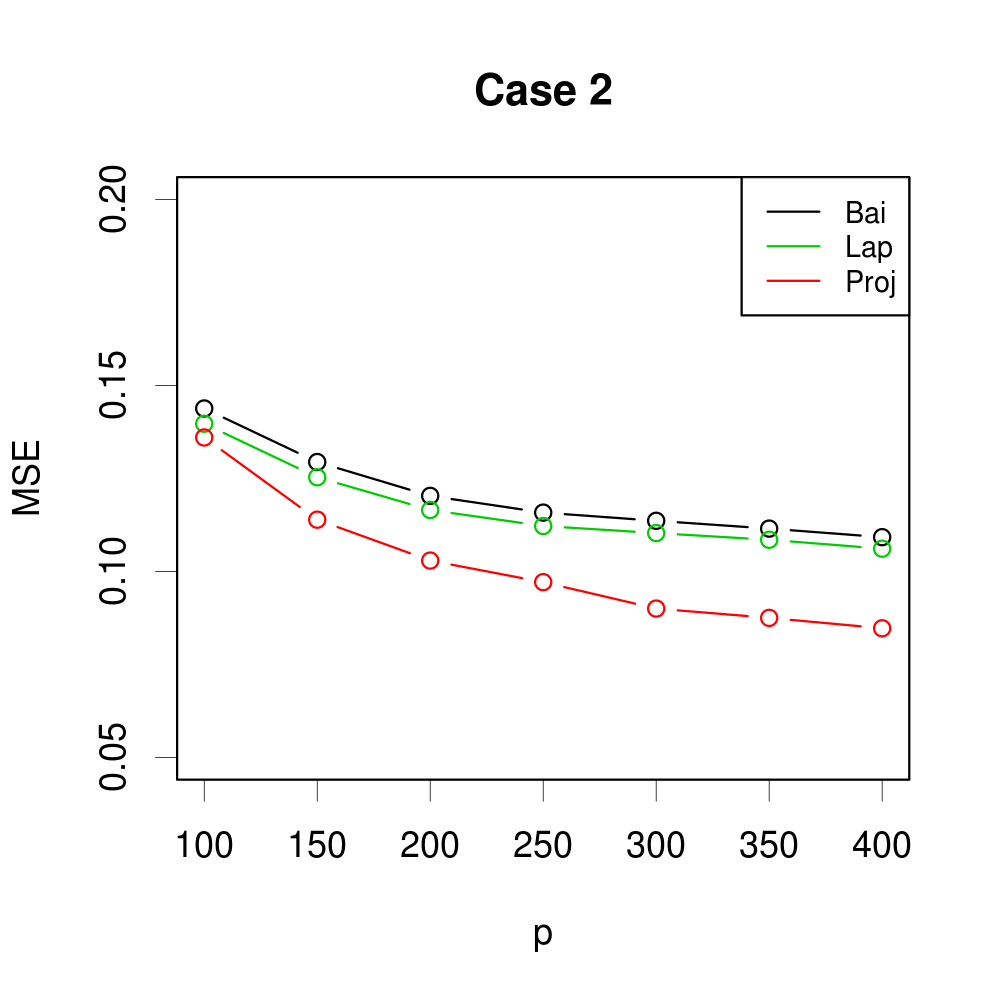}
\end{subfigure}
	\begin{subfigure}{.24\textwidth}
		\centering
		\includegraphics[width=4.3cm,height=4.3cm]{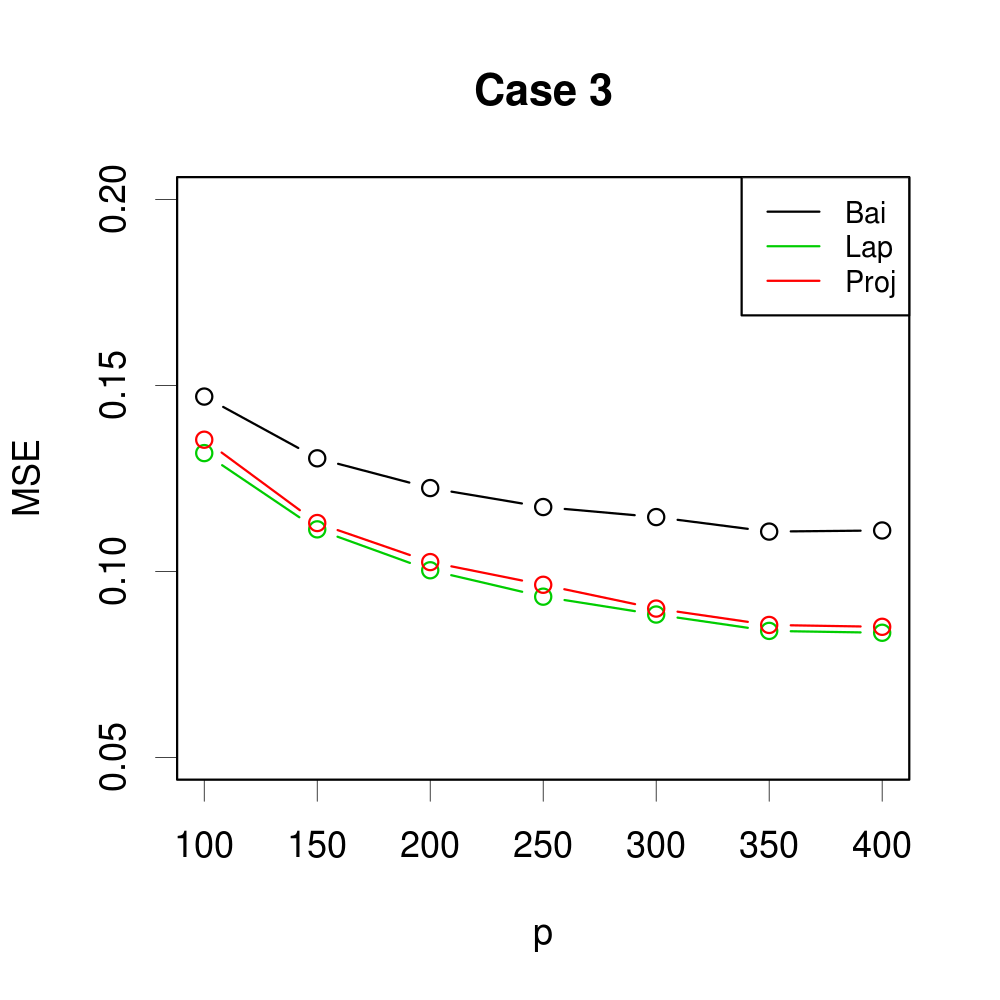}
	\end{subfigure}
	\begin{subfigure}{.24\textwidth}
	\centering
	\includegraphics[width=4.3cm,height=4.3cm]{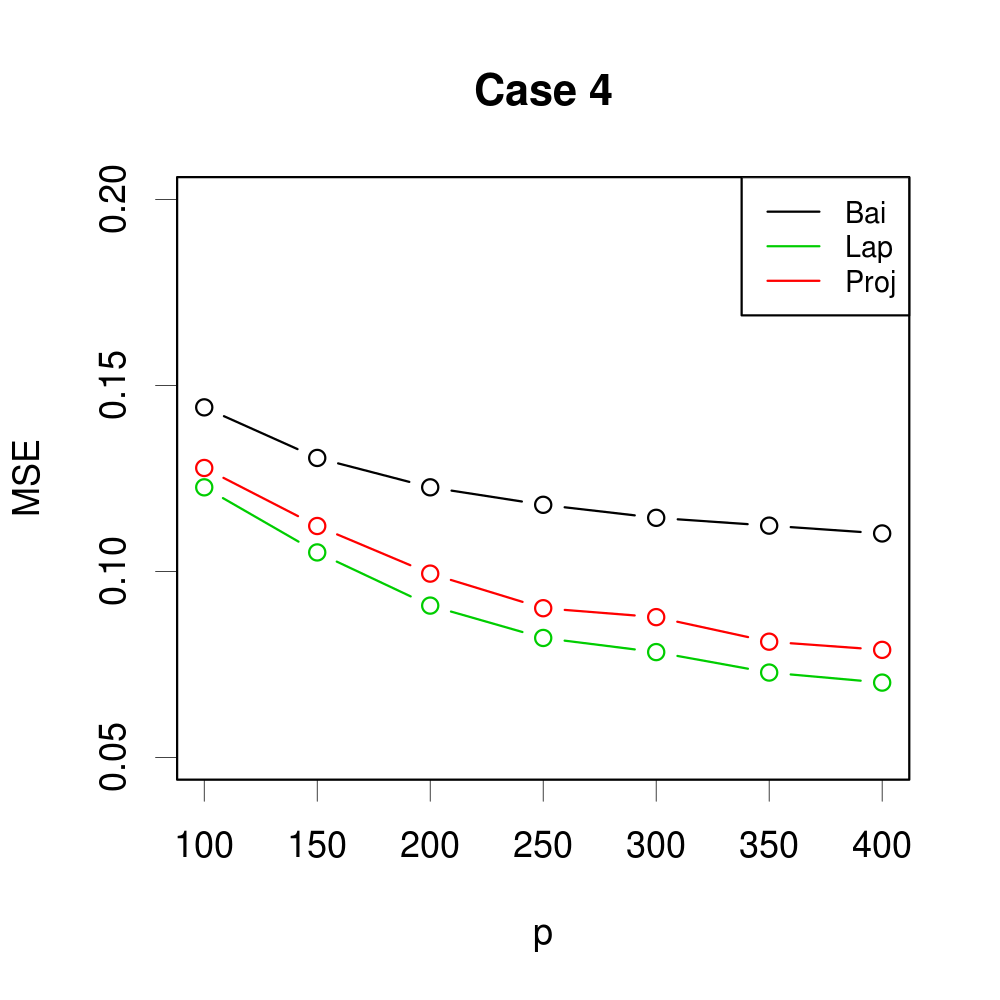}
\end{subfigure}
	\caption{Average errors for estimated common components under 4 cases with 500 replications. ``Bai" (black) is for the PCA method in \cite{bai2003inferential}, ``Lap" (green) is for the adaptive approach with Laplacian penalty, and ``Proj" (red) is for Projection penalty.}
	\label{fig2}
\end{figure}

Figure \ref{fig2} illustrates the  adaptivity and advantages of the new approaches. For Case 1 where the network is totally irrelevant, the proposed methods still give as good results as conventional PCA solution by automatically controlling the penalization. Actually, even in this case both the Laplacian  and Projection penalized methods perform slightly better than PCA method. In Case 2, the Projection penalized method performs the best due to the correct specification of the loading space. The Laplacian penalized method is not as promising as the Projection penalized method  because the assumption of small $\text{tr}(\Bb^\top\mathcal{L}_n\Bb)$ is violated in this case, though it still outperforms the conventional PCA method. In Case 3, both penalized methods lead to relatively smaller and similar estimation errors. This is because the nonzeros eigenvalues $\tau_j$ decrease smoothly, and the Laplacian penalty doesn't benefit a lot from the exactly weighting. We discuss more details about this case in  Section B of the supplementary materials. In Case 4, large eigenvalues of $\mathcal{L}_n$ are significantly distinguished from the small ones, so the Laplacian based method performs the best. It can be concluded from Figure \ref{fig2} that the new approaches are adaptive to multiple network structures. It's worth applying the penalized methods in real application, as one only sacrifices a bit more computation for the potential impressive improvements.

\subsection{Determine the number of factors}
	\begin{table*}[hbpt]
	\begin{center}
		\addtolength{\tabcolsep}{15pt}
		\caption{Selection of $r$ in 500 replications. ``ER" is for the eigenvalue ratio method  in \cite{ahn2013eigenvalue}, ``Lap" and ``Proj" are our `` one step further" approaches with corresponding penalties. $T=50$ and $r=3$.}\label{tab1}
		\renewcommand{\arraystretch}{1}
		\scalebox{1}{ 		\begin{tabular*}{15cm}{lllll}
				\toprule[1.2pt]
Case&$p$&ER&Lap&Proj\\\hline
\multirow{7}*{Case1}&100&2.526(170$|$7)&2.534(169$|$9)&2.534(169$|$9)
\\
&150&2.740(92$|$0)&2.738(92$|$0)&2.740(91$|$0)
\\
&200&2.816(67$|$0)&2.818(66$|$0)&2.820(65$|$0)
\\
&250&2.864(52$|$0)&2.864(52$|$0)&2.868(51$|$0)
\\
&300&2.910(33$|$0)&2.910(33$|$0)&2.910(33$|$0)
\\
&350&2.928(29$|$0)&2.928(29$|$0)&2.928(29$|$0)
\\
&400&2.936(23$|$0)&2.938(22$|$0)&2.938(22$|$0)
\\\hline
\multirow{7}*{Case2}&100&2.802(74$|$4)&2.796(76$|$3)&2.894(44$|$4)
\\
&150&2.884(45$|$0)&2.880(47$|$0)&2.970(12$|$0)
\\
&200&2.912(35$|$0)&2.912(35$|$0)&2.976(10$|$0)
\\
&250&2.946(24$|$0)&2.948(23$|$0)&2.986(7$|$0)
\\
&300&2.950(18$|$0)&2.950(18$|$0)&2.990(3$|$0)
\\
&350&2.960(17$|$0)&2.962(16$|$0)&3.000(0$|$0)
\\
&400&2.964(12$|$0)&2.966(11$|$0)&3.000(0$|$0)
\\\hline
\multirow{7}*{Case3}&100&2.774(86$|$3)&2.906(36$|$1)&2.892(39$|$1)
\\
&150&2.868(46$|$0)&2.968(12$|$0)&2.968(12$|$0)
\\
&200&2.904(36$|$1)&2.984(7$|$0)&2.976(10$|$0)
\\
&250&2.932(28$|$0)&2.984(7$|$0)&2.982(8$|$0)
\\
&300&2.946(22$|$0)&3.000(0$|$0)&2.998(1$|$0)
\\
&350&2.964(12$|$0)&2.994(2$|$0)&2.994(2$|$0)
\\
&400&2.946(22$|$0)&2.996(2$|$0)&2.996(2$|$0)
\\\hline
\multirow{7}*{Case4}&100&2.762(90$|$3)&2.880(48$|$1)&2.878(50$|$2)\\
&150&2.880(46$|$0)&2.968(11$|$0)&2.960(14$|$0)
\\
&200&2.898(39$|$0)&2.982(7$|$0)&2.976(10$|$0)
\\
&250&2.918(31$|$0)&2.992(4$|$0)&2.986(7$|$0)
\\
&300&2.964(16$|$0)&3.000(0$|$0)&3.000(0$|$0)
\\
&350&2.944(20$|$0)&3.000(0$|$0)&2.998(1$|$0)
\\
&400&2.970(12$|$0)&2.998(1$|$0)&2.994(3$|$0)\\
				\bottomrule[1.2pt]		
		\end{tabular*}}

	\end{center}
\end{table*}
We still use Case 1 to Case 4 to empirically study the performances of the ``one step further'' approach introduced in Section \ref{sec:onestepfurther},  except that we set $\sigma_e^2=4$. It makes the problem more challenging due to the low signal-to-noise ratio. We compare the ``one step further" approach (both penalties are considered) with ``ER" criterion.  For all the methods, $k_{\max}$ is always set as 10. We report the results ($T=50$) in Table  \ref{tab1}  by the form $a(b|c)$, where $a$ is the sample mean of the estimated $r$ in 500 replications, $b$ and $c$ are the number of underestimation and overestimation respectively. The  $T=20$ case is presented in the supplementary materials Section C.

The results in Table \ref{tab1} match with our findings in Figure \ref{fig2}. The estimators perform similarly in Case 1 where the network and loadings are independent, while for the other three cases, the ``one step further'' approach  estimates $r$  more accurately. The improvements will be more impressive for  the large $p$ small $T$ scenarios with $T=20$.  Hence, it is  convincing that our method  is more reliable when the network information is  correct. Meanwhile, it's sufficiently adaptive and safe to handle the case with misleading network structure.

\section{Real data analysis}\label{sec7}
We collect weekly returns (calculated by adjusted close price) of companies composing the Standard and Poor's 100 (S\&P100) index  from 2016-01-08 to 2019-01-01. The dataset is downloaded from Yahoo Finance. The Fox Corporation and Alphabet Inc. have two classes of stocks, while the returns of  Berkshire Hathaway are not completely downloaded and thus removed, so the raw dataset is a $157\times 101$ panel. We standardize these time series separately, then use the \textsf{R} package \texttt{factorcpt} in \cite{barigozzi2018simultaneous} to detect change points. It turns out the weeks containing 2016-07-15 and 2018-01-19 might be two change points for the dataset. We only keep the stock returns in 2016 and 2017 for further analysis,  thus  a $104\times 101$ panel remains. We construct the network by sectors, i.e., two companies are linked if and only if they are in the same sector. The sector information is collected with function \texttt{tq\_index} in \textsf{R} package \texttt{tidyquant}.

\subsection{Factor number}
The first step is to specify the number of common factors.   The ``ER"  criterion and our ``one step further" approach with both penalties lead to the same results that  $r=1$. In Figure \ref{fig3}, we plot the leading 20 eigenvalues of the matrix $(pT)^{-1}\Xb\Db^{-1}\Xb^\top$ and the corresponding explained variability, where $\Db=\Ib_p$ for ``ER", $\Db=\Ib_p+\alpha\Ub_1\Ub_1^\top$ for Projection penalty ($\alpha$ and $m$ are selected by $C_L$ with $r=1$). The Laplacian penalty leads to almost the same results as Projection penalty and  is not presented in this figure.
\begin{figure}[hbpt]
	\begin{subfigure}{.5\textwidth}
		\centering
		\includegraphics[width=7.5cm,height=7cm]{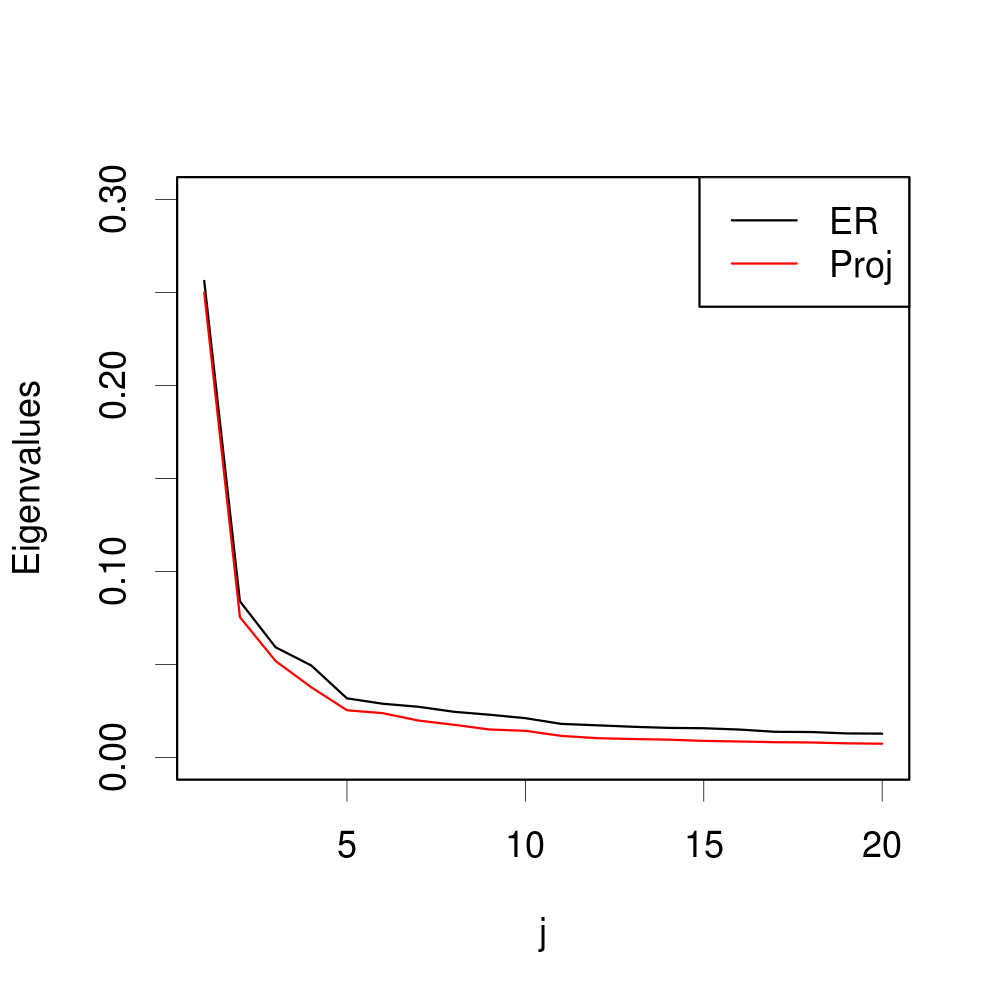}
		\caption{Eigenvalues before and after shrinkage}
	\end{subfigure}
	\begin{subfigure}{.5\textwidth}
		\centering
		\includegraphics[width=7.5cm,height=7cm]{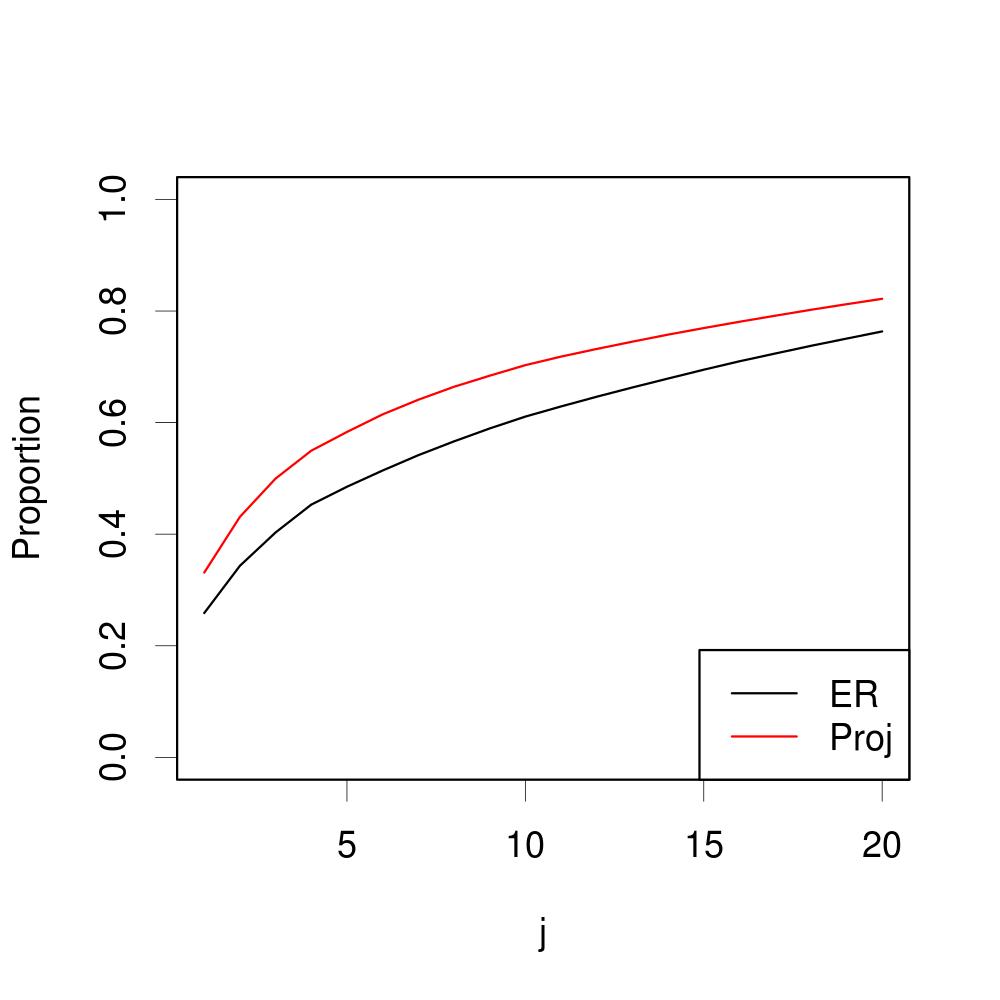}
		\caption{Explained variability}
	\end{subfigure}
	\caption{Eigenvalues for S\&P100 weekly stock return data}
	\label{fig3}
\end{figure}

Figure \ref{fig3} shows that there is a huge gap between the first and second eigenvalues both before and after shrinkage. It implies the existence of a strong factor, which might be the market premium in CAPM model (\cite{ross1977capital}). However, without shrinkage the first eigenvalue can only explain 25.9\% of the variability for this large dataset. Motivated by the Fama-French 5 factor model, we decide to take $r=5$, while without penalty the leading 5 eigenvalues can explain 48.5\% of the variability. The proportion goes up to 58.3\% with the penalty due to the shrinkage of errors, which increases the importance of  factors. We are also interested in the selected $m$ and $\alpha$. By taking $r=5$, applying $C_L$ again leads to $m=11$, which is exactly the total number of sectors for S\&P100 components. It also suggests $\alpha\approx 0.25$ for both Laplacian and Projection penalties, which implies the network can partially explain the interconnectivity of the variables.

\subsection{Estimation and interpretation of factors}
We then apply our methods and also the PCA solution in \cite{bai2003inferential} to estimate the loadings and factor scores given $r=5$, $\alpha=0.25$ and $m=11$. In Figure \ref{fig4}, we plot the estimated loadings (after varimax rotation) of two sectors, Consumer Discretionary and Energy, while the others are presented in Section D of the supplementary materials. As we expect, the new approach forces the loadings to be more ``similar" if they are in the same sector.
\begin{figure}[hbpt]
	\begin{subfigure}{.5\textwidth}
		\centering
		\includegraphics[width=7.5cm,height=5cm]{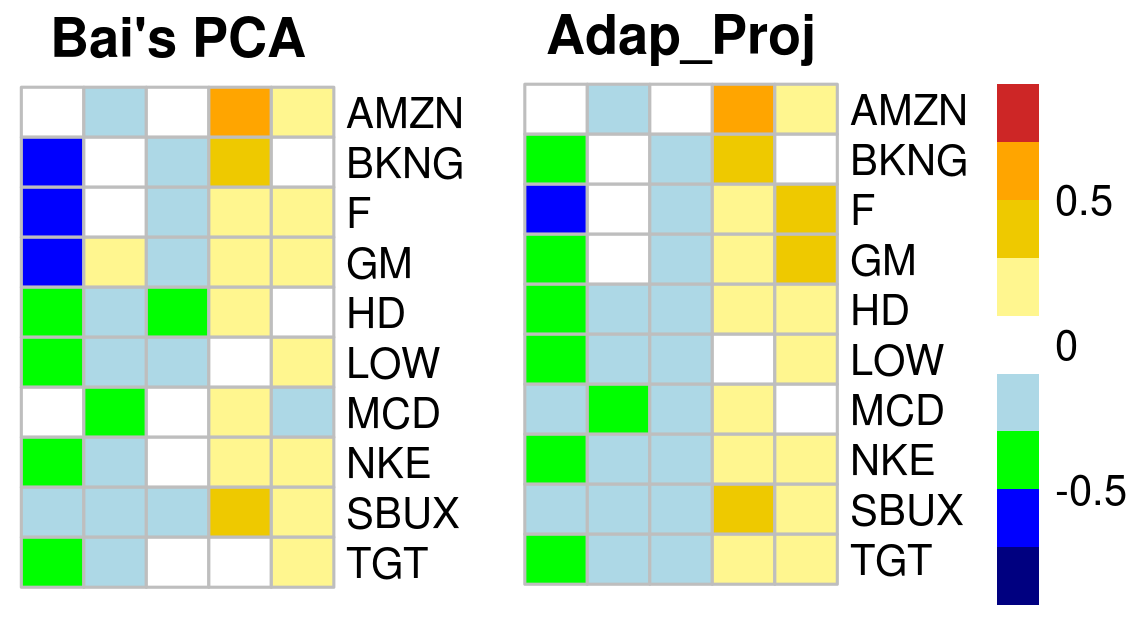}
		\caption{Consumer Discretionary}
	\end{subfigure}
	\begin{subfigure}{.5\textwidth}
		\centering
		\includegraphics[width=7.5cm,height=4cm]{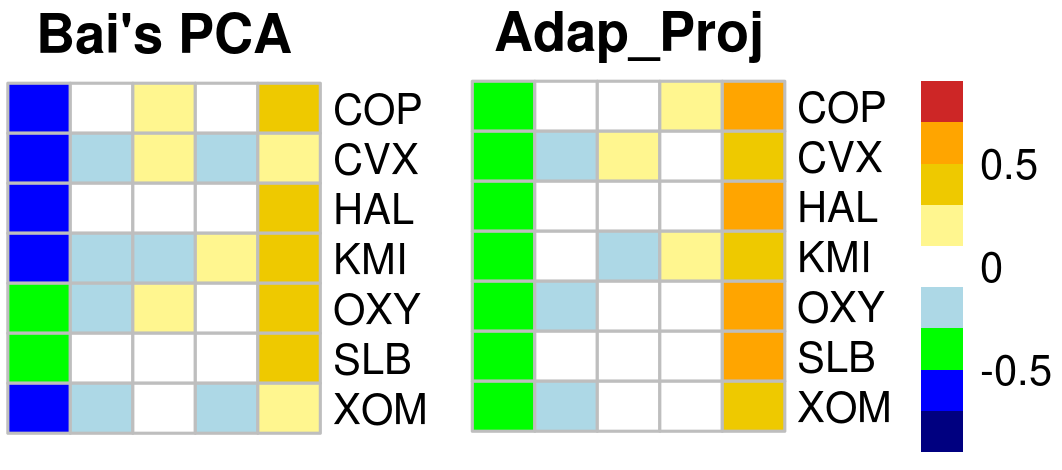}
		\caption{Energy}
	\end{subfigure}
	\caption{Heatmap of the estimated loadings corresponding to two sectors: Consumer Discretionary and Energy. ``Adap\_Proj" is for the proposed Projection based approach, while the Laplacian based shows very close results.}
	\label{fig4}
\end{figure}

In classical factor analysis,  the common factors are usually defined based on nonzero loadings. However, in the large-dimensional cases, each column of the estimated loading matrix can contain a large number of nonzero entries, making it more challenging to define and interpret the corresponding factor. For the penalized methods, the factors can be more easily defined at a ``sector" level instead of ``company" level, because the estimated loadings in the same sector are close to each other.  For example, the Projection based approach in Figure \ref{fig4} shows more clearly that the sector ``Energy" is driven by  factor 1 and factor 5. It's not necessary to exactly specify which companies are driven by the factors. Note that the priori network can be more general than a grouping structure and in such cases we could first do clustering based on the adjacency matrix.

\subsection{Comparison with PCA solution }
It's hard to evaluate the performances of these approaches with real data because the true common components are inaccessible. To compare them, the following  recursive validation procedure is adopted. For any week $t$ in the year 2017, we recursively  estimate a loading matrix $\hat\Bb_t$ using the nearest 52 observations before  $t$ (a $52\times 101$ panel), and get $\hat\bF_t$ by linear regression of $\bx_t$ on $\hat\Bb_t$. In each step, we record the mean squared error $p^{-1}\|\bx_t-\Bb_t\bF_t\|^2$ and the coefficient of determination R$^2$. Then, the average error (Ave\_MSE) and average R$^2$ (Ave\_R$^2$)  are calculated to compare these methods. We also expect a better approach should be more robust and lead to smaller variability for the estimated loading matrix  in the recursive procedure. The variability (Var\_B) is approximatively calculated  by ${1}/{52}\sum_{t=1}^{52}(pr)^{-1}\|\hat\Bb_t-\hat\Bb_{t-1}\|_F^2$. Furthermore, inspired by our $C_L$ criterion,  the accuracy of estimated common components can also be roughly compared by
\[
\frac{1}{pT}\|\Fb\Bb^\top-\hat\Fb\hat\Bb\|_F^2\approx \frac{1}{pT}\|\Xb-\hat\Fb\hat\Bb\|_F^2-\hat\sigma_e^2+\frac{2r\hat\sigma_e^2}{pT}\text{tr}(\Db^{-1}).
\]
 We refer this to the adjusted error (Adj\_error). Table \ref{tab2} summarizes the results,  while Figure \ref{fig3} shows the  R$^2$ and variability of $\Bb$ (calculated by $\frac{1}{pr}\|\hat\Bb_t-\hat\Bb_{t-1}\|_F^2$) at each step $t$.
	\begin{table*}[hbpt]
	\begin{center}
		\addtolength{\tabcolsep}{2pt}
		\caption{Comparison of estimations on S\&P100 weekly stock returns}\label{tab2}
		\renewcommand{\arraystretch}{1.2}
		\scalebox{1}{ 		\begin{tabular*}{10cm}{ccccc}
				
				\toprule[1.2pt]
				Methods&Adj\_error&Ave\_MSE&Var\_B&Ave\_R$^2$\\\hline
				Bai&0.0490&0.5533&0.1596&0.1694
				\\
				Lap&0.0449&0.5438&0.1446&0.1830
				\\
				Proj&0.0339&0.5441&0.0726&0.1826\\
				\bottomrule[1.2pt]		
		\end{tabular*}}
		
	\end{center}
\end{table*}
\begin{figure}[hbpt]
	\begin{subfigure}{0.45\textwidth}
		\centering
		\includegraphics[width=7.5cm,height=5cm]{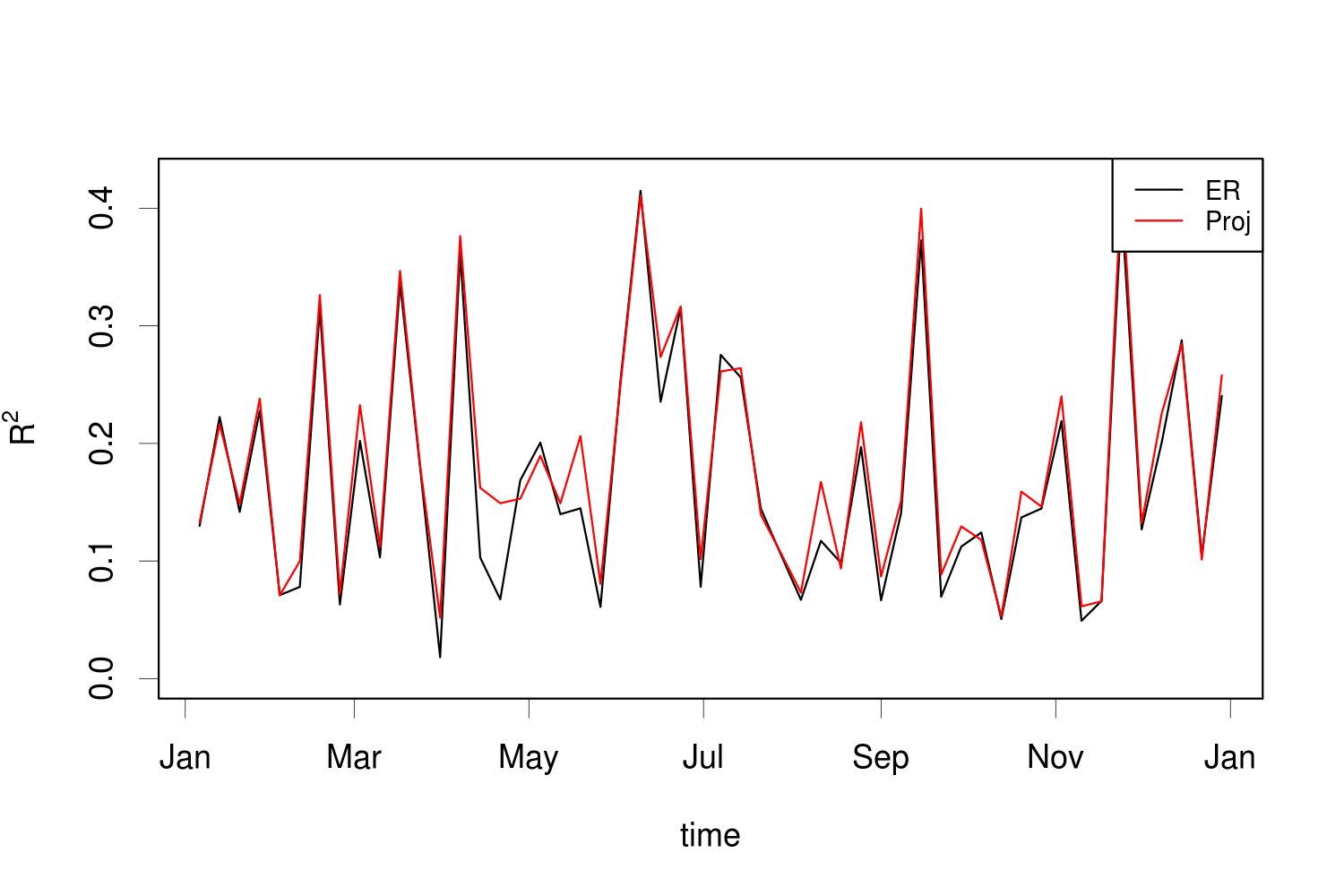}
		\caption{R$^2$}
	\end{subfigure}
	\begin{subfigure}{0.45\textwidth}
		\centering
		\includegraphics[width=7.5cm,height=5cm]{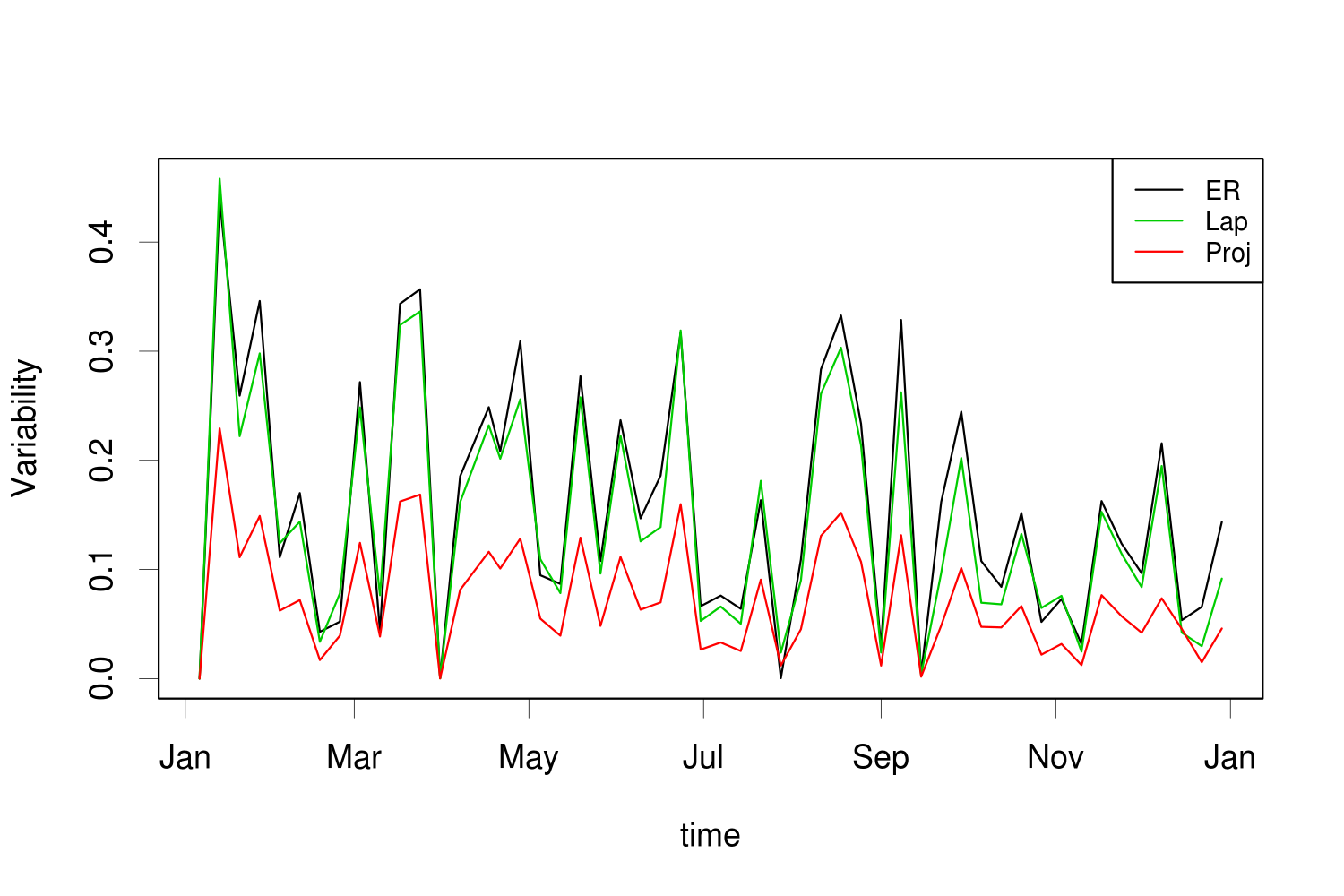}
		\caption{Variability of $\hat\Bb_t$}
	\end{subfigure}
	\caption{The results of the recursive validation process. R$^2$ of Laplacian based method shows nearly the same results as Projection based and is omitted.}
	\label{fig5}
\end{figure}

From  Table \ref{tab2} and Figure \ref{fig5}, we can see that the adaptive methods achieve lower validation errors and higher R$^2$, which indicates the estimated loading matrix and factor scores are more accurate. In addition, the Projection penalty makes the loading matrix much more stable compared with its competitors. To summarize, we conclude that the public companies in S\&P100 index are connected  by the sector network. Our methods contributes to the extraction of common factors and specification of the factor structure by adaptively utilizing the network information.

\section{Conclusions and discussions}\label{sec8}
In this paper,  we assume a priori network is observed in advance to characterize interconnectivity of large scale variables in large-dimensional  approximate factor models. We propose adaptive estimations with Laplacian penalty and Projection penalty based on the network information, and provide closed-form solutions.  Consistency and asymptotic  normality  are studied under very mild conditions, which concludes the new methods can lead to more accurate estimations due to the trade-off between bias and variance. A computationally efficient $C_L$ criterion and a ``one step further" approach are proposed to select tuning parameters and adaptively  determine the number of common factors respectively.

The network linked  framework in this paper can be further extended as follows. Firstly, it may be interesting to combine the Laplacian penalty and Projection penalty by setting $\Db=\Ib_p+\Ub_1\bLambda_1\Ub_1^\top$ for the estimation procedure, where $\Ub_1$, $\bLambda_1$ are the leading $(p-m)$ eigenvector matrix and corresponding eigenvalue matrix. This new penalty is in essence still projection-based, and it  leads to very similar results as Projection penalty for the simulation scenarios, so we do not present it in the paper. Secondly, the network can also be applied to describe the connection of mean vector (if the mean is assumed unknown but not necessarily 0 in the factor model) or the covariance matrix of idiosyncratic errors.  In addition, if the considered factor model are with matrix-value observations as in \cite{wang2019factor}, we can use two network structures to separately regularize  corresponding loading matrices.

On the other hand, the proposed adaptive approach  facilitates many existing factor-adjusted statistical learning problems. For instance, it can be applied to the covariance matrix estimation in \cite{fan2011high} and \cite{fan2013large}, the screening procedure in \cite{wang2012factor}, or the change point localization in \cite{baltagi2017identification}, with more accurately estimated loadings and factor scores. Such extensions deserve separate study and are left as our future work.

\section*{Acknowledgements}
 Long Yu's research is partially supported by China Scholarship Council (No.201806100081). Yong He's research is partially supported by the grant of the National Science Foundation of China (NSFC 11801316),  Natural Science Foundation
of Shandong Province (ZR2019QA002) and National Statistical Scientific Research Project (2018LY63).  Xinsheng Zhang's research is partially supported by the grant of the National Science Foundation of China (NSFC 11571080).

\bibliographystyle{plain}
\bibliography{newref}

	\begin{appendices}
	
\section*{Appendix}

The supplementary document is organized as follows. In section \ref{seca}, we give the detailed proofs of all the main theorems. That is,  the consistency, asymptotic normality, and also the verification of Assumption E.  We design a special case in section \ref{secb}, where the theoretical optimal tuning parameters for Laplacian penalty can also be derived, so we can compare the theoretical results of the proposed two penalized methods. The case also explains the simulation results of case 3 in section 6.2 of the main paper. In section \ref{secc}, we present Tables \ref{tab1}-\ref{tab3} to illustrate detailed simulation results, which supplement section 6 of the main paper. In section \ref{secd}, we present the estimated loadings of the real dataset.

\section{Proofs of Main Theorems}\label{seca}

\begin{lemma}\label{lema1}
	Under the error  Assumption D, for any $p\times p$ symmetric matrix $\Qb$ with $\|\Qb\|\le O(1)$, we have
	\[
	\mathbb{E}|\bepsilon_t^\top\Qb\bepsilon_t-\text{tr}(\Qb)|^2\le (M+2)\|\Qb\|_F^2 ,\quad\mathbb{E}(\bepsilon_t^\top\Qb\bepsilon_s)^2=\|\Qb\|_F^2.
	\]
\end{lemma}
\begin{proof}
	By $\bepsilon_t^\top\Qb\bepsilon_t=\sum_{i,j}Q_{ij}\epsilon_{ti}\epsilon_{tj}$, we have
	\[
	\begin{split}
	\mathbb{E}(\bepsilon_t^\top\Qb\bepsilon_t)=&\sum_{i,j}Q_{ij}\mathbb{E}(\epsilon_{ti}\epsilon_{tj})=\text{tr}(\Qb)\\
	\mathbb{E}(\bepsilon_t^\top\Qb\bepsilon_t)^2=&\sum_{i,j}\sum_{u,v}Q_{ij}Q_{uv}\mathbb{E}(\epsilon_{ti}\epsilon_{tj}\epsilon_{tu}\epsilon_{tv})\\
	=&\sum_{i=j=u=v}Q_{ii}^2\mathbb{E}(\epsilon_{ti}^4)+\sum_{i=j,u=v,i\ne u}Q_{ii}Q_{uu}+\sum_{i=u,j=v,i\ne j}Q_{ij}^2+\sum_{i=v,j=u,i\ne j}Q_{ij}^2\\
	=&\sum_iQ_{ii}^2\bigg(\mathbb{E}(\epsilon_{ti}^4)-1\bigg)+(\sum_{i}Q_{ii})^2+2\sum_{i\ne j}Q_{ij}^2.
	\end{split}
	\]
	Therefore,
	\[
		\mathbb{E}|\bepsilon_t^\top\Qb\bepsilon_t-\text{tr}(\Qb)|^2=	\mathbb{E}(\bepsilon_t^\top\Qb\bepsilon_t)^2-\text{tr}^2(\Qb)\le (M+2)\|\Qb\|_F^2.
	\]
	Similarly for $t\ne s$,
	\[
\begin{split}
\mathbb{E}(\bepsilon_t^\top\Qb\bepsilon_s)^2=&\sum_{i,j}\sum_{u,v}Q_{ij}Q_{uv}\mathbb{E}(\epsilon_{ti}\epsilon_{sj}\epsilon_{tu}\epsilon_{sv})=\sum_{i=u,j=v}Q_{ij}Q_{uv}\mathbb{E}(\epsilon_{ti}\epsilon_{sj}\epsilon_{tu}\epsilon_{sv})=\sum_{i,j}Q_{ij}^2=\|\Qb\|_F^2.
\end{split}
\]
	which concludes the lemma.
\end{proof}

\begin{lemma}\label{lema2}
	Under Assumptions B and D, for any $p\times p$ symmetric matrix $\Qb$ that  $\|\Qb\|\le O(1)$, $T\times T$ matrix $\Pb$ that  $\|\Pb\|\le O(1)$ and $t\le T$, for sufficiently large $p,T$, we have
	\[
	\mathbb{E}\bigg\|\bepsilon_t^\top\Qb\mathcal{E}^\top\Pb\Fb-\text{tr}(\Qb)\bp_{t}^\top\Fb\bigg\|^2\le M^2T\|\Qb\|_F^2.
	\]
	where $\bp_t$ is the $t$-th row of $\Pb$.
\end{lemma}	
\begin{proof}
	Firstly,
	\[
	\bepsilon_t^\top\Qb\mathcal{E}^\top\Pb\Fb=\bepsilon_t^\top\Qb(\sum_{s=1}^{T}\bepsilon_s\bp_s^\top\Fb)=\bepsilon_t^\top\Qb\bepsilon_t\bp_t^\top\Fb+\bepsilon_t^\top\Qb(\sum_{s\ne t}\bepsilon_s\bp_s^\top\Fb).
	\]
	By lemma \ref{lema1}, Assumptions B and D,
	\[
	\mathbb{E}\bigg\|\bepsilon_t^\top\Qb\bepsilon_t\bp_t^\top\Fb-\text{tr}(\Qb)\bp_{t}^\top\Fb\bigg\|^2\le M(M+2)\|\Qb\|_F^2\|\bp_t\|^2
	\]
	On the other hand,
	\[
	\bzeta_t:=\bepsilon_t^\top\Qb(\sum_{s\ne t}\bepsilon_s\bp_s^\top\Fb)=\sum_{i,j}Q_{ij}\epsilon_{ti}\sum_{s\ne t}\epsilon_{sj}\bp_s^\top\Fb.
	\]
	Similarly to lemma \ref{lema1}, it's not difficult  to verify $\mathbb{E}(\bzeta_t)={\bf 0}$ and
	\[
	\begin{split}
	\mathbb{E}(\|\bzeta_t\|^2)=&\sum_{i,j}\sum_{u,v}\sum_{s\ne t}\sum_{w\ne t}Q_{ij}Q_{uv}\mathbb{E}(\epsilon_{ti}\epsilon_{sj}\epsilon_{tu}\epsilon_{wv}\bp_s^\top\Fb\bp_{w}^\top\Fb)\\
	=&\sum_{i,j,s\ne t}Q_{ij}^2\mathbb{E}\bigg(\epsilon_{ti}^2\epsilon_{sj}^2(\bp_s^\top\Fb)^2\bigg)\le M\|\Qb\|_F^2\sum_{s}\|\bp_s\|^2\le M^2T\|\Qb\|_F^2.
	\end{split}
	\]
	Combine the above results to conclude the lemma.
\end{proof}

\begin{lemma}\label{lema3}
	Under Assumptions B and D, for any vector $\bu $ and $T\times T$ matrix $\Pb$ that $\|\bu\|^2=1$  and $\|\Pb\|\le O(1)$, we have
	\[
	\mathbb{E}\|\bu^\top\mathcal{E}^\top\Pb\Fb\|^2\le M^2T.
	\]
\end{lemma}	
\begin{proof}
	Define $\eta_t=\bu^\top\bepsilon_t$, then $\mathbb{E}\eta_t=0$, $\mathbb{E}\eta_s\eta_t=0$ and $\mathbb{E}\eta_t^2=1$. Define $\tilde\Fb=\Pb\Fb$, then $\bu^\top\mathcal{E}^\top\Pb\Fb=\sum_{t=1}^{T}\eta_t\tilde\bF_t$. Hence,
	\[
	\begin{split}
	\mathbb{E}\|\bu^\top\mathcal{E}^\top\Pb\Fb\|^2=&\sum_{k=1}^{r}\sum_{t=1}^{T}\sum_{s=1}^{T}\mathbb{E}(\eta_t\eta_s\tilde f_{tk}\tilde f_{sk})=\sum_{k=1}^{r}\sum_{t=1}^{T}\mathbb{E}(\eta_t^2\tilde f_{tk}^2)=\mathbb{E}\|\Pb\Fb\|_F^2\le M\|\Pb\|_F^2\le M^2T.
	\end{split}
	\]
	which concludes the lemma.
\end{proof}

\begin{lemma}\label{lema4}
	Under Assumptions B and D, for any $p\times p$ matrix $\Qb$ that $\|\Qb\|\le O(1)$ and $T\times T$ matrix $\Pb$ that $\|\Pb\|\le O(1)$, we have
	\[
\mathbb{E}\|\Qb^\top\mathcal{E}^\top\Pb\Fb\|_F^2\le M^2T\|\Qb\|_F^2.
	\]
\end{lemma}	
\begin{proof}
	Define the $j$-th column of $\Qb$ as $\bq_j$, by lemma \ref{lema3}, we have
	\[
	\mathbb{E}\bigg(\|\Qb^\top\mathcal{E}^\top\Pb\Fb\|_F^2\bigg)\le\mathbb{E}\bigg(\sum_{j=1}^{p}\|\bq_j\|^2\bigg\|\frac{\bq_j^\top}{\|\bq_j\|}\mathcal{E}^\top\Pb\Fb\bigg\|^2\bigg)\le M^2T\|\Qb\|_F^2,
	\]
	and the lemma holds.
\end{proof}

\begin{lemma}\label{lema5}
	Define $\bLambda_r$  as the diagonal matrix composed of the leading $r$ eigenvalues of $(pT)^{-1}\Xb\Db^{-1}\Xb^\top$, then under Assumptions A-E, as $p,T\rightarrow \infty$, $\bLambda_{r,jj}$ are always positive and distinct,
	\[
	\bLambda_{r,jj}\asymp 1, \quad\text{for}\quad j=1,\ldots,r.
	\]
\end{lemma}	
\begin{proof}
	Denote ${\bf S}=p^{-1}\Bb^\top\Db^{-1}\Bb$, then
\[
\frac{1}{pT}\Xb\Db^{-1}\Xb^\top=\frac{1}{T}\Fb{\bf S}\Fb^\top+\frac{1}{pT}(\Fb\Bb^\top\Db^{-1}\Eb^\top+\Eb\Db^{-1}\Bb\Fb^\top+\Eb\Db^{-1}\Eb^\top).
\]
Note that $T^{-1}\Fb{\bf S}\Fb^\top$ shares the same non-zero eigenvalues with ${\bf S}$ because of $T^{-1}\Fb^\top\Fb=\Ib_r$ a.s., while by Assumption E, $\lambda_j({\bf S})$ are distinct and positive. Meanwhile, $\|{\bf S}\|\le \lambda_1(\bSigma_B)$.

For the remaining terms, we have
\[
\begin{split}
\bigg\|\frac{1}{pT}\Fb\Bb^\top\Db^{-1}\Eb^\top\bigg\|^2\le& \bigg\|\frac{1}{pT}\Fb\Bb^\top\Db^{-1}\Eb^\top\bigg\|_F^2\le\frac{1}{p^2T}\text{tr}(\mathcal{E}\Pb_2\Db^{-1}\Bb\Bb^\top\Db^{-1}\Pb_2^\top\mathcal{E}^\top)\|\Pb_1\|^2\\
\le&\frac{M^2}{p^2T}\sum_{t=1}^{T}\bepsilon_t^\top\Pb_2\Db^{-1}\Bb\Bb^\top\Db^{-1}\Pb_2^\top\bepsilon_t\\
=&O_p\bigg(\frac{M^2}{p^2}\|\Bb\|_F^2\|\Db^{-1}\|^2\|\Pb_2\|^2\bigg)(\text{by Lemma \ref{lema1}})\\
=&O_p(p^{-1}).
\end{split}
\]
Similarly, we have $\|(pT)^{-1}\Eb\Db^{-1}\Bb\Fb^\top\|^2=O_p(p^{-1})$. By Lemma \ref{lema1} we also have
\begin{equation}\label{equc1}
\frac{1}{pT}\Eb\Db^{-1}\Eb^\top=\frac{1}{pT}\Pb_1\bigg(\text{tr}(\Pb_2\Db^{-1}\Pb_2^\top)\Ib_T+\Rb\bigg)\Pb_1^\top,
\end{equation}
where $\Rb$ is $T\times T$ with $\mathbb{E}(\mathrm{R}^2_{ts})\le (M+2)\|\Pb_2\Db^{-1}\Pb_2\|_F^2$ for any $s,t\le T$. Therefore,
\[
\bigg\|\frac{1}{pT}\Eb\Db^{-1}\Eb^\top\bigg\|\le\frac{\|\Pb_1\|^2\|\Pb_2\|^2\text{tr}(\Db^{-1})}{pT}+O_p(p^{-1/2})=O_p(T^{-1}+p^{-1/2}).
\]
The lemma holds with  Weyl's inequality and Assumption E.
\end{proof}

\begin{lemma}\label{lema6}
Under Assumptions A-E, there exists a sequence of invertible matrices $\Hb$ (dependent on $p,T$ and tuning parameters) such that
\[
\frac{1}{T}\|\hat\Fb-\Fb\Hb\|_F^2=O_p\bigg(\frac{1}{p}+\frac{\text{tr}(\Db^{-2})}{pT^2}\bigg).
\]
\end{lemma}	
\begin{proof}
	Follow Lemma \ref{lema5}, by the definition of $\hat \Fb$,  we have
	\[
	\hat\Fb\bLambda_r=\frac{1}{pT}\Xb\Db^{-1}\Xb^\top\hat\Fb=\frac{1}{pT}\bigg(p\Fb{\bf S}\Fb^\top+\Fb\Bb^\top\Db^{-1}\Eb^\top+\Eb\Db^{-1}\Bb\Fb^\top+\Eb\Db^{-1}\Eb^\top\bigg)\hat\Fb.
	\]
	Define $\Hb={\bf S}(\Fb^\top\hat\Fb/T)\bLambda_r^{-1}$, then $\Hb=O_p(1)$ and
	\begin{equation}\label{equc2}
	(\hat\Fb-\Fb\Hb)\bLambda_r=\frac{1}{pT}(\Fb\Bb^\top\Db^{-1}\Eb^\top+\Eb\Db^{-1}\Bb\Fb^\top+\Eb\Db^{-1}\Eb^\top)\hat\Fb.
	\end{equation}
	We first give some coarse bounds for the right-hand terms. That is,
	\[
	\bigg\|\frac{1}{pT}\Fb\Bb^\top\Db^{-1}\Eb^\top\hat\Fb\bigg\|_F^2\le\frac{1}{p^2T^2}\|\Fb\|_F^2\|\Bb^\top\Db^{-1}\Eb^\top\|_F^2\|\hat\Fb\|_F^2=\frac{r^2}{p^2}\|\Bb^\top\Db^{-1}\Eb^\top\|_F^2.
	\]
	Hence by Lemma \ref{lema1}, the above term is  bounded by $O_p(T/p)$. And by a similar procedure, it's easy that  $\|(pT)^{-1}\Eb\Db^{-1}\Bb\Fb^\top\hat\Fb\|_F^2=O_p(T/p)$. Next by equation (\ref{equc1}), we have
	\[
	\|(pT)^{-1}\Eb\Db^{-1}\Eb^\top\hat\Fb\|_F^2=O_p(T^{-1}+T/p).
	\]
	Combining the above bounds leads to $T^{-1}\|\hat\Fb-\Fb\Hb\|_F^2=O_p(T^{-2}+p^{-1})$ and $\Hb$ is invertible..
	
	Now we can come back to equation (\ref{equc2}) and construct more subtle bounds. Firstly.
	\[
	\begin{split}
	\bigg\|\frac{1}{pT}\Fb\Bb^\top\Db^{-1}\Eb^\top\hat\Fb\bigg\|_F^2=&\frac{1}{p^2T}\|\Bb^\top\Db^{-1}\Eb^\top(\hat\Fb-\Fb\Hb+\Fb\Hb)\|_F^2\\
	\le&\frac{2}{p^2T}\bigg(\|\Bb^\top\Db^{-1}\Eb^\top\|_F^2\|\hat\Fb-\Fb\Hb\|_F^2+\|\Bb^\top\Db^{-1}\Eb^\top\Fb\Hb\|_F^2\bigg)\\
	=&O_p(p^{-1})+\frac{1}{p^2T}O_p(\|\Bb^\top\Db^{-1}\Eb^\top\Fb\|_F^2)=O_p(p^{-1}),
	\end{split}
	\]
	where we use  Lemma \ref{lema3} in the last line. Next, we have
	\[
	\begin{split}
	\bigg\|\frac{1}{pT}\Eb\Db^{-1}\Eb^\top\hat\Fb\bigg\|_F^2=&O_p\bigg(\frac{1}{p^2T^2}\|\Eb\Db^{-1}\Eb^\top\Fb\|_F^2+\bigg\|\frac{1}{pT}\Eb\Db^{-1}\Eb^\top(\hat\Fb-\Fb\Hb)\bigg\|_F^2\bigg)\\
	=&O_p\bigg(\frac{\text{tr}^2(\Db^{-1})}{p^2T}+\frac{\|\Db^{-1}\|_F^2}{p^2}+\frac{\text{tr}^2(\Db^{-1})}{p^2}\|T^{-1}(\hat\Fb-\Fb\Hb)\|_F^2+\frac{\|\Rb\|_F^2}{p^2}\|T^{-1}(\hat\Fb-\Fb\Hb)\|_F^2\bigg)\\
	=&O_p\bigg(\frac{\text{tr}(\Db^{-2})}{pT}+\frac{1}{p}\bigg),
	\end{split}
	\]
	where the second line holds because of Lemma \ref{lema2} and equation (\ref{equc1}). 	Note that we use $p^{-2}\text{tr}^2(\Db^{-1})\le p^{-1}\text{tr}(\Db^{-2})$ for the consistency rates. Hence, combined with $\|(pT)^{-1}\Eb\Db^{-1}\Bb\Fb^\top\hat\Fb\|_F^2=O_p(T/p)$, now we have
	\[
	\frac{1}{T}\|\hat\Fb-\Fb\Hb\|_F^2=O_p\bigg(\frac{1}{p}+\frac{\text{tr}(\Db^{-2})}{pT^2}\bigg),
	\]
	and the lemma holds.
\end{proof}

\begin{lemma}\label{lema7}
		Under the same conditions and notations in Lemma \ref{lema6},  we have
		\[
\frac{1}{T}\Fb^\top(\hat\Fb-\Fb\Hb)=O_p(T^{-1}), 		\frac{1}{T}\hat\Fb^\top(\hat\Fb-\Fb\Hb)=O_p(T^{-1}).
		\]
	\end{lemma}
\begin{proof}
	The results hold with \ref{equc2} by using similar technique as in Lemma \ref{lema6}.
\end{proof}

\begin{lemma}\label{lema8}
	Assume ${\bf S}$ has the spectral decomposition ${\bf S}=\bGamma_S\bLambda_S\bGamma_S^\top$, then under the same conditions and notations in Lemma \ref{lema6},  we have
	\[
	\Hb^\top\Hb\stackrel{p}{\rightarrow}\Ib_r,\quad \frac{1}{T}\Fb^\top\hat\Fb\stackrel{p}{\rightarrow}\bGamma_S.
	\]
\end{lemma}
\begin{proof}
	Lemma \ref{lema7} implies $\Hb^\top\Hb=\Ib_r+O_p(T^{-1})$. By the definition $\Hb={\bf S}(\Fb^\top\hat\Fb/T)\bLambda_r^{-1}$, we have
	\[
	\bGamma_S^\top\Hb\bLambda_r=\bLambda_S\bGamma_S^\top\Hb+O_p(T^{-1}).
	\]
	Note that $\bLambda_r=\bLambda_S+O_p(T^{-1}+p^{-1/2})$, and  $\lambda_j({\bf S})$ are distinct by Assumption E, so
	\[
	\bGamma_S^\top\Hb=c\Ib_r+O_p(T^{-1}+p^{-1/2}),\quad\text{where}\quad c=\pm1.
	\]
	Without loss of generality we can assume $c=1$, hence $T^{-1}\Fb^\top\hat\Fb\stackrel{p}{\rightarrow}\bGamma_S$.
\end{proof}

\begin{lemma}\label{lema9}
		Under the same conditions and notations in Lemma \ref{lema6},  we have
		\[
	\frac{1}{p}\bigg\|\frac{1}{T}\Eb^\top(\hat\Fb-\Fb\Hb)\bigg\|_F^2=O_p\bigg(\frac{1}{pT}+\frac{\text{tr}(\Db^{-2})}{pT^3}\bigg).
		\]
	\end{lemma}
\begin{proof}
	By equation (\ref{equc2}),
	\[
	T^{-1}\Eb^\top(\hat\Fb-\Fb\Hb)\bLambda_r=\frac{1}{pT^2}\Eb^\top(\Fb\Bb^\top\Db^{-1}\Eb^\top+\Eb\Db^{-1}\Bb\Fb^\top+\Eb\Db^{-1}\Eb^\top)\hat\Fb.
	\]
	It's not hard to verify
	\[
	\bigg\|\Eb^\top\Fb\Bb^\top\Db^{-1}\Eb^\top\hat\Fb\bigg\|_F^2\le \|\Eb^\top\Fb\|^2\bigg\|\Bb^\top\Db^{-1}\Eb^\top\hat\Fb\bigg\|_F^2\le O_p(p^2T^2).
	\]
	while for the second term,
	\[
	\begin{split}
	\bigg\|\Eb^\top\Eb\Db^{-1}\Bb\Fb^\top\hat\Fb\bigg\|_F^2\le& \bigg\|\Pb_2^\top\mathcal{E}^\top\Pb_1^\top\Pb_1\mathcal{E}\Pb_2\Db^{-1}\Bb\bigg\|_F^2\|\hat\Fb^\top\Fb\|_F^2.
	\end{split}
	\]
	By similar technique in Lemma \ref{lema2} we will have
	\[
\mathcal{E}^\top\Pb_1^\top\Pb_1\mathcal{E}\Pb_2\Db^{-1}\Bb=\text{tr}(\Pb_1^\top\Pb_1)\Pb_2\Db^{-1}\Bb+\Rb_2,
	\]
	where $\Rb_2$ is $T\times r$ with $\mathrm{R}_{2,tk}=O_p(\sqrt{pT})$. Hence
	\[
	\bigg\|\Eb^\top\Eb\Db^{-1}\Bb\Fb^\top\hat\Fb\bigg\|_F^2\le O_p\bigg(pT^4\bigg).
	\]
	Further,
	\[
\|\Eb^\top\Eb\Db^{-1}\Eb^\top\Fb\|_F=	\bigg\|\Eb^\top\Pb_1\bigg(\text{tr}(\Pb_2\Db^{-1}\Pb_2^\top)\Pb_1^\top\Fb+\Rb_3\bigg)\bigg\|_F\le O_p\bigg(\sqrt{pT}\text{tr}(\Db^{-1})+\sqrt{pT^3}\|\Db^{-1}\|_F\bigg),
	\]
	where $\Rb_3$ is $T\times r$ matrix with $|\mathrm{R}_{3,tk}|=O_p(\sqrt{T}\|\Db^{-1}\|_F)$ by Lemma \ref{lema2}. Hence
	\[
	\frac{1}{p}\bigg\|\frac{1}{T}\Eb^\top(\hat\Fb-\Fb\Hb)\bigg\|_F^2\le O_p\bigg(\frac{1}{pT^2}+\frac{1}{p^2}+\frac{\text{tr}^2(\Db^{-1})}{p^2T^3}+\frac{\|\Db^{-1}\|_F^2}{p^2T}\bigg)\le O_p\bigg(\frac{1}{pT}+\frac{\text{tr}(\Db^{-2})}{pT^3}\bigg).
	\]
	Note that we use $p^{-2}\text{tr}^2(\Db^{-1})\le p^{-1}\text{tr}(\Db^{-2})$ again.
\end{proof}

\begin{lemma}\label{lema10}
	Under the same conditions and notations in Lemma \ref{lema6},  we have
	\[
	\frac{1}{p}\|\hat\Bb-\Bb\Hb \|_F^2=O_p\bigg(\frac{1}{p}\|(\Db^{-1}-\Ib_p)\Bb\|_F^2+\frac{\text{tr}(\Db^{-2})}{pT}\bigg).
	\]
\end{lemma}	
\begin{proof}
By the solution of $\hat\Bb$,
\[
\begin{split}
&\hat\Bb-\Bb\Hb=\frac{1}{T}\Db^{-1}(\Bb\Fb^\top+\Eb^\top)(\hat\Fb-\Fb\Hb+\Fb\Hb)-\Bb\Hb\\
=&(\Db^{-1}-\Ib_p)\Bb\Hb+\frac{1}{T}\Db^{-1}\Eb^\top\Fb\Hb+\frac{1}{T}\Db^{-1}\Bb\Fb^\top(\hat\Fb-\Fb\Hb)+\frac{1}{T}\Db^{-1}\Eb^\top(\hat\Fb-\Fb\Hb).
\end{split}
\]
By Lemma \ref{lema4},
\[
\mathbb{E}\bigg\|\frac{1}{T}\Db^{-1}\Eb^\top\Fb\bigg\|_F^2=O\bigg(\frac{\text{tr}(\Db^{-2})}{T}\bigg).
\]
Then, combined with Lemmas \ref{lema7} and \ref{lema9},
\[
\frac{1}{p}\|\hat\Bb-\Bb\Hb\|_F^2=O_p\bigg(\frac{1}{p}\|(\Db^{-1}-\Ib_p)\Bb\|_F^2+\frac{\text{tr}(\Db^{-2})}{pT}\bigg),
\]
and the lemma holds.
\end{proof}

\begin{lemma}\label{lema11}
		Under the same conditions and notations in Lemma \ref{lema6},  we have
	\[
\frac{1}{pT}\|\hat\Cb-\Cb\|_F^2=O_p\bigg(\frac{1}{p}\|(\Db^{-1}-\Ib_p)\Bb\|_F^2+\frac{\text{tr}(\Db^{-2})}{pT}+\frac{1}{p}\bigg).
	\]
	\end{lemma}
\begin{proof}
	
	For common components,
	\[
	\begin{split}
	\frac{1}{pT}\|\hat\Cb-\Cb\|_F^2=&\frac{1}{pT}\|\hat\Bb\Hb^{-1}\Hb\hat\Fb^\top-\Bb\Fb^\top\|_F^2\\
	\le& \frac{1}{pT}\|\hat\Bb\Hb^{-1}(\hat\Fb\Hb^\top-\Fb)^\top\|_F^2+\frac{1}{pT}\|(\hat\Bb\Hb^{-1}-\Bb)\Fb^\top\|_F^2\\
	\le&O_p\bigg(\frac{1}{p}\|(\Db^{-1}-\Ib_p)\Bb\|_F^2+\frac{\text{tr}(\Db^{-2})}{pT}+\frac{1}{p}\bigg).
	\end{split}
	\]
	The last line comes from Lemmas \ref{lema6} and \ref{lema10}.
\end{proof}
\textbf{Theorem 4.1}. 	Under Assumptions A-E, there exists a sequence of invertible matrices $\Hb$ (dependent on $p,T$ and tuning parameters) such that $\Hb^\top\Hb\stackrel{p}{\rightarrow}\Ib_r$, and
\[
\begin{split}
\frac{1}{T}\|\hat\Fb-\Fb\Hb\|_F^2=&O_p\bigg(\frac{1}{p}+\frac{\text{tr}(\Db^{-2})}{pT^2}\bigg),\\
\frac{1}{p}\|\hat\Bb-\Bb\Hb \|_F^2=&O_p\bigg(\frac{1}{p}\|(\Db^{-1}-\Ib_p)\Bb\|_F^2+\frac{\text{tr}(\Db^{-2})}{pT}\bigg),\\
\frac{1}{pT}\|\hat\Cb-\Cb\|_F^2=&O_p\bigg(\frac{1}{p}\|(\Db^{-1}-\Ib_p)\Bb\|_F^2+\frac{\text{tr}(\Db^{-2})}{pT}+\frac{1}{p}\bigg),
\end{split}
\]
where $\hat\Bb=T^{-1}\Db^{-1}\Xb^\top\hat\Fb$, $\hat\Fb$ is $\sqrt{T}$ times the leading $r$ eigenvectors of $(pT)^{-1}\Xb\Db^{-1}\Xb^\top$ and $\hat\Cb=\hat\Fb\hat\Bb^\top$. $\Db=\Db_1:=\Ib_p+\alpha\mathcal{L}_n$ for Laplacian penalty, while $\Db=\Db_2:=\Ib_p+\alpha\Ub_1\Ub_1^\top$ for Projection penalty with $\Ub_1=(\bu_1,\ldots,\bu_{p-m})$, where $\bu_j$ is the $j$-th eigenvector of $\mathcal{L}_n$.
\begin{proof}
	The theorem follows directly Lemmas \ref{lema7}-\ref{lema11}.
\end{proof}

\textbf{Corollary 4.1}. 	Given $m$, Assumption E always holds with $\alpha=p/(T\|\tilde\Bb_1\|_F^2)$ for Projection penalty as long as Assumption C holds.
\begin{proof}
	On one hand, if $p^{-1}\|\tilde\Bb_1\|_F^2=o(1)$, for Projection penalty we will have
	\[
	{\bf S}=\frac{1}{p}\Bb^\top\Db_{2}^{-1}\Bb=\frac{1}{p}\Bb^\top\Bb-\frac{\alpha}{1+\alpha}\frac{1}{p}\tilde\Bb_1^\top\tilde\Bb_1\rightarrow\bSigma_B.
	\]
	 Hence, Assumption E holds with Assumption C. On the other hand, if $p^{-1}\|\tilde\Bb_1\|_F^2\gtrsim1$, we will have $\alpha=o(1)$ so still ${\bf S}\rightarrow\bSigma_B$.
\end{proof}

\textbf{Corollary 4.2}. 	Assumption E always holds with $\alpha=(T\max_{j}\{\tau_j\|\bb_j\|^2\})^{-1}$ for Laplacian penalty as long as Assumption C holds.
\begin{proof}
	On one hand, if $\max_{j}\{\tau_j\|\bb_j\|^2\}=o(1)$, for Laplacian penalty we have
	\[
	{\bf S}=\frac{1}{p}\Bb^\top\Db_{1}^{-1}\Bb=\frac{1}{p}\Bb^\top\Bb-\frac{1}{p}\Bb^\top(\Ib_p-\Db_1^{-1})\Bb,
	\]
	where
	\[
	\frac{1}{p}\text{tr}(\Bb^\top(\Ib_p-\Db_1^{-1})\Bb)=\frac{1}{p}\sum_{j=1}^{p}\frac{\alpha\tau_j\|\bb_j\|^2}{1+\alpha\tau_j}\le T^{-1}.
	\]
	Hence, ${\bf S\rightarrow\bSigma_B}$. On the other hand, if $\max_{j}\{\tau_j\|\bb_j\|^2\}\gtrsim 1$, we have $\alpha=o(1)$ and again ${\bf S\rightarrow \bSigma_B}$.
\end{proof}

\begin{lemma}\label{lema12}
	Under Assumptions B and D, for any $T$-dimensional vector $\bu$ such that $\|\bu\|=1$, $T\times T$ matrix $\Pb$ such that $\|\Pb\|\le O(1)$ and $p\times p$ symmetric matrix $\Qb$, we have
	\[
	\mathbb{E}\bigg\|\bu^\top\mathcal{E}\Qb\mathcal{E}^\top\Pb\Fb-\text{tr}(\Qb)\bu^\top\Pb\Fb\bigg\|^2=O_p\bigg(T\|\Qb\|_F^2\bigg).
	\]
	\end{lemma}
\begin{proof}
	Firstly,
	\[
	\begin{split}
	\bu^\top\mathcal{E}\Qb\mathcal{E}^\top\Pb\Fb=&\sum_{t=1}^{T}u_t\bepsilon_t^\top\Qb\sum_{s=1}^{T}\bepsilon_s\bp_s^\top\Fb=\sum_{t=1}^{T}u_t\bepsilon_t^\top\Qb\bepsilon_t\bp_t^\top\Fb+\sum_{t=1}^{T}u_t\bepsilon_t^\top\Qb\sum_{s\ne t}\bepsilon_s\bp_s^\top\Fb\\
	=&\mathcal{\uppercase\expandafter{\romannumeral 1}}+\mathcal{\uppercase\expandafter{\romannumeral 2}}.
	\end{split}
	\]
	We have
	\[
	\begin{split}
	\mathbb{E}\bigg\|\mathcal{\uppercase\expandafter{\romannumeral 1}}-\text{tr}(\Qb)\bu^\top\Pb\Fb\bigg\|^2=&\mathbb{E}\bigg\|\sum_{t}u_t\bigg(\sum_{i,j}\epsilon_{ti}Q_{ij}\epsilon_{tj}-\text{tr}(\Qb)\bigg)\bp_t\Fb\bigg\|^2\\
	=&\sum_{t}u_t^2\mathbb{E}\bigg\|\sum_{i,j}\epsilon_{ti}Q_{ij}\epsilon_{tj}-\text{tr}(\Qb)\bigg\|^2\mathbb{E}\|\bp_t^\top\Fb\|^2\\
	\le &M(M+2)\|\Qb\|_F^2\|\Pb\|^2.
	\end{split}
	\]
	On the other hand,
	\[
\begin{split}
\mathbb{E}\|\mathcal{\uppercase\expandafter{\romannumeral 2}}\|^2=&\mathbb{E}\bigg\|\sum_{t}\sum_{s\ne t}\sum_{i,j}u_tQ_{ij}\epsilon_{ti}\epsilon_{sj}\bp_s^\top\Fb\bigg\|^2\\
=&\mathbb{E}\sum_{t_1}\sum_{t_2}\sum_{s_1\ne t_1}\sum_{s_2\ne t_2}\sum_{i_1,i_2,j_1,j_2}u_{t_1}u_{t_2}Q_{i_1,j_1}Q_{i_2,j_2}\epsilon_{t_1,i_1}\epsilon_{s_1,j_1}\epsilon_{t_2,i_2}\epsilon_{s_2,j_2}\bp_{s_1}^\top\Fb\Fb^\top\bp_{s_2}\\
=&\mathbb{E}\sum_{t_1=t_2}\sum_{s_1= s_2,s_1\ne t_1}\sum_{i_1=i_2,j_1=j_2}u_{t_1}^2Q_{i_1,j_1}^2\|\bp_{s_1}^\top\Fb\|^2\\
&+\mathbb{E}\sum_{t_1=s_2}\sum_{s_1= t_2,s_1\ne t_1}\sum_{i_1=j_2,i_2=j_1}u_{t_1}u_{t_2}Q_{i_1,j_1}^2\bp_{s_1}^\top\Fb\Fb^\top\bp_{s_2}\\
\le &O_p\bigg(T\|\Qb\|_F^2\bigg).
\end{split}
\]
Combining the results to conclude the lemma.
\end{proof}

\textbf{Theorem 4.2}.
	Denote the spectral decomposition ${\bf S}=\bGamma_S\bLambda_S\bGamma_S^\top$. 	When Assumptions A-E and F1 hold, we have for the estimated factor scores,
	\begin{enumerate}
		\item if $\|\Db^{-1}\|_F/T=o(1)$,
		\[
		\sqrt{p}(\hat\bF_t-\Hb^\top\bF_t)\stackrel{d}{\rightarrow}\mathcal{N}({\bf 0}, \bLambda_S^{-1}\bGamma_S^\top\Vb_t\bGamma_S\bLambda_S^{-1}),
		\]
		where $\Hb$ is the same as in Theorem 4.1 and $\Vb_t$ is defined in Assumption F1. It can be shown $	\Vb_t=\|\bp_{1,t}\|^2\lim_{p,T\rightarrow\infty}p^{-1}(\Bb^\top\Db^{-1}\Pb_2^\top\Pb_2\Db^{-1}\Bb)$.
		\item 	 Otherwise if $\|\Db^{-1}\|_F/T\ge O(1)$, $\hat\bF_t-\Hb^\top\bF_t=O_p(T^{-1}\text{tr}(\Db^{-1}))$.
	\end{enumerate}
\begin{proof}
	By equation (\ref{equc2}), we have for any $t\le T$,
	\[
	\hat\bF_t-\Hb\bF_t=\frac{1}{pT}\bLambda_r^{-1}\hat\Fb^\top(\Eb\Db^{-1}\Bb\bF_t+\Fb\Bb^\top\Db^{-1}\be_t+\Eb\Db^{-1}\be_t).
	\]
	Applying similar techniques as in Lemma \ref{lema6}, we have
	\[
	\|\hat\Fb^\top\Eb\Db^{-1}\Bb\bF_t\|^2\le O_p(pT), \quad \|\hat\Fb^\top\Fb\Bb^\top\Db^{-1}\be_t\|^2\le O_p(pT^2),
	\]
	while by Lemma \ref{lema12}, we have
	\[
	\|\hat\Fb^\top\Eb\Db^{-1}\be_t\|^2=O_p\bigg(T\|\Db^{-1}\|_F^2+\text{tr}^2(\Db^{-1})\bigg)\le O_p\bigg(p\|\Db^{-1}\|_F^2\bigg).
	\]
	Hence,
	\[
	\sqrt{p}(\hat\bF_t-\Hb^\top\bF_t)=\bLambda_r^{-1}\frac{\hat\Fb^\top\Fb}{T}\frac{1}{\sqrt{p}}\Bb^\top\Db^{-1}\be_t+O_p\bigg(\frac{\|\Db^{-1}\|_F}{T}\bigg)+o_p(1).
	\]
	By previous theorem, we have known that $\bLambda_r\stackrel{p}{\rightarrow}\bLambda_S$ and $T^{-1}\hat\Fb^\top\Fb\stackrel{p}{\rightarrow}\bGamma_S^\top$, so combined with Assumption F1, when $\|\Db^{-1}\|_F/T=o(1)$,
	\[
	\sqrt{p}(\hat\bF_t-\Hb^\top\bF_t)\stackrel{d}{\rightarrow}\mathcal{N}({\bf 0}, \bLambda_S^{-1}\bGamma_S^\top\Vb_t\bGamma_S\bLambda_S^{-1}),
	\]
	where $\Vb_t=\lim_{p,T\rightarrow \infty}p^{-1}\text{cov}(\Bb^\top\Db^{-1}\be_t)$. Otherwise, $\hat\bF_t-\Hb^\top\bF_t=O_p(\text{tr}(\Db^{-1})/(pT))$. Further by Lemma \ref{lema1},
	\[
	\Vb_t=\|\bp_{1,t}\|^2\lim_{p,T\rightarrow\infty}p^{-1}(\Bb^\top\Db^{-1}\Pb_2^\top\Pb_2\Db^{-1}\Bb),
	\]
	and the theorem holds.
\end{proof}

\textbf{Theorem 4.3}.
	Under Assumptions A-E and F2, use the same notations as in Theorem 4.2, then  for the estimated loadings we have
	\[
	\frac{\sqrt{T}(\hat\bb_j-\Hb^\top\Bb^\top\bd_j^{-1})}{\|\bd_j^{-1}\|}\stackrel{d}{\rightarrow}\mathcal{N}({\bf 0},\bGamma_S^\top\Wb_j\bGamma_S),
	\]
	where $\Wb_j$ is defined in Assumption F2 and
	\[
	\Wb_j=\lim_{p,T\rightarrow \infty}\bigg\|\frac{\Pb_2\bd_j^{-1}}{\|\bd_j^{-1}\|}\bigg\|^2\frac{1}{T}\mathbb{E}(\Fb^\top\Pb_1\Pb_1^\top\Fb).
	\]
\begin{proof}
By the proof of Lemma \ref{lema10}, for the loadings, we have
\[
\sqrt{T}(\hat\bb_j-\Hb^\top\Bb^\top\bd_j^{-1})=\frac{1}{\sqrt{T}}\Hb^\top\Fb^\top\Eb\bd_j^{-1}+\frac{1}{\sqrt{T}}(\hat\Fb-\Fb\Hb)^\top\Fb\Bb^\top\bd_j^{-1}+\frac{1}{\sqrt{T}}(\hat\Fb-\Fb\Hb)^\top\Eb\bd_j^{-1}.
\]
By Lemma \ref{lema7},
\[
\frac{1}{\sqrt{T}}(\hat\Fb-\Fb\Hb)^\top\Fb\Bb^\top\bd_j^{-1}=O_p\bigg(T^{-1/2}\bigg),
\]
while a similar proof as Lemma \ref{lema8} shows that
\[
\frac{1}{\sqrt{T}}(\hat\Fb-\Fb\Hb)^\top\Eb\bd_j^{-1}=o_p(1).
\]
Therefore,
	\[
	\sqrt{T}(\hat\bb_j-\Hb^\top\Bb^\top\bd_j^{-1})=\frac{1}{\sqrt{T}}\Hb^\top\Fb^\top\Eb\bd_j^{-1}+o_p(1),
	\]
	where $\bd_j^{-1}$ is the $j$-th row of the matrix $\Db^{-1}$.
	By Assumption F2 and $\Hb\stackrel{p}{\rightarrow}\bGamma_S$,
	\[
\frac{\sqrt{T}(\hat\bb_j-\Hb^\top\Bb^\top\bd_j^{-1})}{\|\bd_j^{-1}\|}\stackrel{d}{\rightarrow}\mathcal{N}({\bf 0},\bGamma_S^\top\Wb_j\bGamma_S),
	\]
	where $\Wb_j=\lim_{p,T\rightarrow\infty}T^{-1}\text{cov}(\Fb^\top\Eb\bd_j^{-1}/\|\bd_j^{-1}\|)$. Further, it's easy that
	\[
	\Wb_j=\lim_{p,T\rightarrow \infty}\bigg\|\frac{\Pb_2\bd_j}{\|\bd_j^{-1}\|}\bigg\|^2\frac{1}{T}\mathbb{E}(\Fb^\top\Pb_1\Pb_1^\top\Fb),
	\]
	since $\Fb$ and $\Eb$ are independent.
\end{proof}

\section{An example to compare two penalties}\label{secb}

As claimed in section 4 of the main paper, generally it's hard to determine the optimal tuning parameter $\alpha$ for Laplacian penalty, which makes it challenging to directly compare the proposed two penalties. In this section a special case is designed so that the two penalties are comparable. It is also adapted into the simulation part as simulated case 3.

Suppose the $p$ variables are separated into $q$ groups, with size $p\theta_k$ for group $k$, where  $1>\theta_1\ge \cdots \ge \theta_q>0$ are positive constants and define $\bar \theta=q^{-1}\sum\theta_k$. We construct the network by linking any pairs in the same group, so that after rearranging, the adjacency matrix is blocked diagonal with each block fully-connected. Under the group structure, the eigenvalues of $\mathcal{L}_n$ are
\[
(\underbrace{\frac{\theta_1}{\bar\theta},\ldots,\frac{\theta_1}{\bar\theta}}_{p\theta_1},\ldots,\underbrace{\frac{\theta_q}{\bar\theta},\ldots,\frac{\theta_q}{\bar\theta}}_{p\theta_q},\underbrace{0,\ldots,0}_{q}),
\]
so exactly $q$ eigenvalues are 0. We further assume $\tau_j\|\tilde \bb_{j}\|^2=z$ for all $j\le p-q$ and $\Pb_1$ and $\Pb_2$ are identical matrices so that $e_{tj}$ are independently and identically distributed. Based on the results in section 4, the estimated errors for loadings are usually larger than the factor scores, so we compare the proposed two penalties only by the approximate MSE in equation (4.3), which is
\[
\text{MSE}\approx \frac{1}{p}\|\Bb^\top(\Ib_p-\Db^{-1})\|_F^2+\frac{1}{pT}\|\Db^{-1}\|_F^2,
\]
where $\Db=\Db_1$ for Laplacian penalty and $\Db=\Db_2$ for Projection penalty.

Firstly, for the Laplacian penalty, we have
\[
\text{MSE}_{lap}\approx\frac{1}{p}\sum_{j=1}^{p}\frac{\alpha^2\tau_j^2\|\tilde \bb_{j}\|^2}{(1+\alpha\tau_j)^2}+\frac{1}{pT}\sum_{j=1}^{p}\frac{1}{(1+\alpha\tau_j)^2}.
\]
By equation (4.2) and $\tau_j\|\bb_j\|^2=z$ for all $j\le p$,  we should take $\alpha=(Tz)^{-1}$ to minimize the error. Define $w_j=\tau_j^{-1}$ for $j\le p-q$, $w_j=0$ for $p-q<j\le p$ and $\bar w=p^{-1}\sum_{j=1}^p w_j$, so that
\[
\text{MSE}_{lap}\approx\frac{1}{pT}\sum_{j=1}^{p}\frac{1}{1+\alpha\tau_j}=\frac{1}{pT}\sum_{j=1}^{p}\frac{w_j}{w_j+(Tz)^{-1}}+\frac{q}{pT}.
\]
For the Projection penalty, it's natural to take $m=q$ in this case, and $\alpha=p/(T\|\tilde\Bb_1\|_F^2)$ by the calculations in section 4. Therefore,
\[
\text{MSE}_{proj}\approx\frac{1}{T(1+\alpha)}+\frac{\alpha^2}{(1+\alpha)^2}\frac{q}{pT}=\frac{1}{T}\frac{\bar w}{\bar w+(Tz)^{-1}}+\frac{\alpha^2}{(1+\alpha)^2}\frac{q}{pT}.
\]
In many cases, we may assume $q/(pT)$ is relatively negligible so the first terms dominate for both penalties. The function $f(x)=x/(x+c)$ is concave when $x>0$ for any constant $c>0$, so under this case the Laplacian penalty should outperform the Projection penalty. However, when the network information is correct, $(Tz)^{-1}$ should be a small number, and the difference between two penalties will not be significant. This is also the reason why the Projection penalty performs nearly as well as the Laplacian penalty under the simulation case 3.

\newpage
\section{Supplementary simulation results}\label{secc}
The attached Table \ref{tab1} contains the detailed simulation results of the estimation errors for common components with $T=50$. Table \ref{tab2} contains the estimation errors for common components with $T=20$. Table \ref{tab3} provides simulated results for number of common factors with $T=20$.

\begin{table*}[hbpt]
	\begin{center}
		\addtolength{\tabcolsep}{7.8pt}
		\caption{Estimation errors of common components, $T=50$.}\label{tab1}
		\renewcommand{\arraystretch}{1}
		\scalebox{1}{ 		\begin{tabular*}{16cm}{cc|ccc|ccc}
				\toprule[1pt]
				\multirow{2}*{Case}&\multirow{2}*{p}&\multicolumn{3}{c|}{Average Error}&\multicolumn{3}{c}{Standard Deviation}\\\cline{3-8}
				&&Bai&Lap&Proj&Bai&Lap&Proj
				\\\hline
\multirow{7}*{Case1}&100&0.1458&0.1444&0.1440&0.0174&0.0173&0.0172
\\
&150&0.1296&0.1285&0.1283&0.0138&0.0138&0.0137
\\
&200&0.1207&0.1196&0.1195&0.0122&0.0122&0.0121
\\
&250&0.1162&0.1153&0.1153&0.0112&0.0112&0.0112
\\
&300&0.1141&0.1129&0.1130&0.0111&0.0111&0.0111
\\
&350&0.1118&0.1110&0.1109&0.0111&0.0112&0.0111
\\
&400&0.1094&0.1084&0.1083&0.0108&0.0108&0.0107\\
\hline
\multirow{7}*{Case2}&100&0.1438&0.1397&0.1360&0.0169&0.0167&0.0157
\\
&150&0.1294&0.1253&0.1139&0.0144&0.0136&0.0122
\\
&200&0.1203&0.1165&0.1029&0.0120&0.0114&0.0097
\\
&250&0.1158&0.1122&0.0971&0.0110&0.0105&0.0089
\\
&300&0.1136&0.1103&0.0900&0.0114&0.0111&0.0088
\\
&350&0.1115&0.1085&0.0875&0.0114&0.011&0.0086
\\
&400&0.1092&0.1061&0.0847&0.0105&0.0101&0.0080
\\\hline
\multirow{7}*{Case3}&100&0.1470&0.1318&0.1354&0.0173&0.0153&0.0156
\\
&150&0.1304&0.1113&0.1130&0.01400&0.0118&0.0121
\\
&200&0.1224&0.1003&0.1025&0.0128&0.0110&0.0110
\\
&250&0.1173&0.0932&0.0964&0.0116&0.0093&0.0096\\
&300&0.1146&0.0884&0.0900&0.0109&0.0090&0.0090\\
&350&0.1107&0.0840&0.0856&0.0108&0.0084&0.0084\\
&400&0.1110&0.0835&0.0851&0.0097&0.0080&0.0081
\\\hline
\multirow{7}*{Case4}&100&0.1441&0.1226&0.1278&0.0172&0.0150&0.0152
\\
&150&0.1305&0.1051&0.1122&0.0144&0.0121&0.0126
\\
&200&0.1226&0.0908&0.0994&0.0127&0.0102&0.0106\\
&250&0.1179&0.0821&0.0901&0.0121&0.0089&0.0095
\\
&300&0.1144&0.0783&0.0877&0.0104&0.0080&0.0085
\\
&350&0.1123&0.0728&0.0811&0.0107&0.0079&0.0085
\\
&400&0.1102&0.0701&0.0789&0.0103&0.0077&0.0081\\
				\bottomrule[1.2pt]		
		\end{tabular*}}
		
	\end{center}
\end{table*}
\newpage
\begin{table*}[hbpt]
	\begin{center}
		\addtolength{\tabcolsep}{7.8pt}
		\caption{Estimation errors of common components, $T=20$.}\label{tab2}
		\renewcommand{\arraystretch}{1}
		\scalebox{1}{ 		\begin{tabular*}{16cm}{cc|ccc|ccc}
				\toprule[1pt]
				\multirow{2}*{Case}&\multirow{2}*{p}&\multicolumn{3}{c|}{Average Error}&\multicolumn{3}{c}{Standard Deviation}\\\cline{3-8}
				&&Bai&Lap&Proj&Bai&Lap&Proj
				\\\hline
	\multirow{7}*{Case1}&100&0.3136&0.2936&0.2937&0.0543&0.0496&0.0498
\\
	&150&0.2918&0.2730&0.2731&0.0456&0.0413&0.0414
\\
	&200&0.2847&0.2665&0.2664&0.0411&0.0375&0.0375\\
	&250&0.2781&0.2603&0.2603&0.0390&0.0353&0.0353
\\
	&300&0.2800&0.2615&0.2616&0.0381&0.0343&0.0343
\\
	&350&0.2731&0.2555&0.2555&0.0369&0.0338&0.0338\\
	&400&0.2691&0.2518&0.2518&0.0356&0.0325&0.0325
\\\hline
		\multirow{7}*{Case2}&100&0.3087&0.2876&0.2727&0.0505&0.0457&0.0445
\\
	&150&0.2888&0.2683&0.2192&0.0418&0.0377&0.0322
\\
	&200&0.2821&0.2619&0.2053&0.0414&0.0375&0.0320
\\
	&250&0.2762&0.2565&0.1939&0.0379&0.0344&0.0293\\
	&300&0.2757&0.2558&0.1748&0.0370&0.0337&0.0269
\\
	&350&0.2704&0.2512&0.1689&0.0363&0.0331&0.0276
\\
	&400&0.2674&0.2483&0.1643&0.0346&0.0317&0.0257
\\\hline
		\multirow{7}*{Case3}&100&0.3099&0.2444&0.2556&0.0507&0.0416&0.0429
\\
	&150&0.2920&0.2103&0.2159&0.0434&0.0357&0.0359
\\
	&200&0.2804&0.1858&0.1969&0.0394&0.0313&0.0319
\\
	&250&0.2774&0.1756&0.1886&0.0384&0.0301&0.0303
\\
	&300&0.2756&0.1675&0.1717&0.0364&0.0289&0.0293
\\
	&350&0.2694&0.1601&0.1642&0.0378&0.0315&0.0315
\\
	&400&0.2724&0.1577&0.1628&0.0351&0.0271&0.0274
\\\hline
		\multirow{7}*{Case4}&100&0.3065&0.2398&0.2474&0.0480&0.0418&0.0410
\\
	&150&0.2901&0.1945&0.2124&0.0416&0.0349&0.0343\\
	&200&0.2803&0.1722&0.1900&0.0420&0.0354&0.0336
\\
	&250&0.2742&0.1600&0.1784&0.0348&0.0292&0.0282
\\
	&300&0.2731&0.1510&0.1701&0.0365&0.0308&0.0294
\\
	&350&0.2687&0.1438&0.1629&0.0332&0.0288&0.0287\\
	&400&0.2668&0.1368&0.1575&0.0336&0.0292&0.0281\\
				\bottomrule[1.2pt]		
		\end{tabular*}}
		
	\end{center}
\end{table*}

\newpage
\begin{table*}[hbpt]
	\begin{center}
		\addtolength{\tabcolsep}{15pt}
\caption{Specify the number of common factors, $T=20,r=3$.}\label{tab3}
\renewcommand{\arraystretch}{1}
\scalebox{1}{ 		\begin{tabular*}{16cm}{lllll}
		\toprule[1.2pt]
Case&$p$&ER&Lap&Proj\\\hline
\multirow{7}*{Case1}&100&2.642(311$|$114)&2.668(307$|$116)&2.668(308$|$116)
\\
&150&2.540(315$|$85)&2.556(315$|$88)&2.568(316$|$89)
\\
&200&2.394(331$|$65)&2.410(329$|$66)&2.402(330$|$66)
\\
&250&2.300(329$|$55)&2.322(326$|$57)&2.302(329$|$55)
\\
&300&2.328(316$|$49)&2.350(314$|$51)&2.344(314$|$50)
\\
&350&2.214(336$|$43)&2.224(335$|$44)&2.214(336$|$43)
\\
&400&2.390(320$|$60)&2.376(319$|$59)&2.380(321$|$59)
\\\hline
\multirow{7}*{Case2}&100&2.560(306$|$86)&2.556(305$|$83)&2.658(286$|$89)
\\
&150&2.356(311$|$62)&2.362(311$|$63)&2.518(257$|$51)
\\
&200&2.408(316$|$61)&2.374(315$|$56)&2.572(249$|$46)
\\
&250&2.362(314$|$52)&2.366(311$|$51)&2.526(225$|$26)
\\
&300&2.402(302$|$52)&2.394(301$|$50)&2.702(173$|$26)
\\
&350&2.276(327$|$45)&2.282(324$|$43)&2.612(198$|$26)
\\
&400&2.350(305$|$44)&2.346(303$|$42)&2.702(163$|$20)
\\\hline
\multirow{7}*{Case3}&100&2.646(294$|$108)&2.864(206$|$89)&2.842(220$|$96)
\\
&150&2.402(312$|$67)&2.662(198$|$36)&2.660(199$|$38)
\\
&200&2.380(322$|$67)&2.752(181$|$41)&2.734(185$|$41)
\\
&250&2.294(322$|$53)&2.676(172$|$25)&2.632(181$|$23)
\\
&300&2.268(337$|$46)&2.726(152$|$21)&2.666(174$|$22)
\\
&350&2.238(327$|$43)&2.702(141$|$17)&2.640(159$|$16)
\\
&400&2.240(337$|$43)&2.708(157$|$18)&2.708(158$|$19)
\\\hline
\multirow{7}*{Case4}&100&2.626(299$|$93)&2.748(269$|$95)&2.784(264$|$101)\\
&150&2.336(320$|$64)&2.600(244$|$64)&2.524(258$|$58)
\\
&200&2.208(329$|$47)&2.544(220$|$32)&2.446(246$|$29)
\\
&250&2.376(308$|$51)&2.692(190$|$32)&2.570(227$|$32)
\\
&300&2.374(321$|$53)&2.836(155$|$37)&2.682(190$|$26)
\\
&350&2.346(323$|$49)&2.764(144$|$24)&2.640(173$|$17)
\\
&400&2.274(316$|$37)&2.758(131$|$17)&2.634(168$|$13)\\
				\bottomrule[1.2pt]		
		\end{tabular*}}
		
	\end{center}
\end{table*}

\newpage
\section{Estimated loadings for S\&P100 dataset}\label{secd}
Here are the estimated loadings for main sectors of S\&P100 dataset, except two sectors presented in Section 7.2 of the main paper and two sectors only containing one company each..
\begin{figure}[hbpt]
	\begin{subfigure}{.3\textwidth}
	\centering
	\includegraphics[width=5cm,height=3cm]{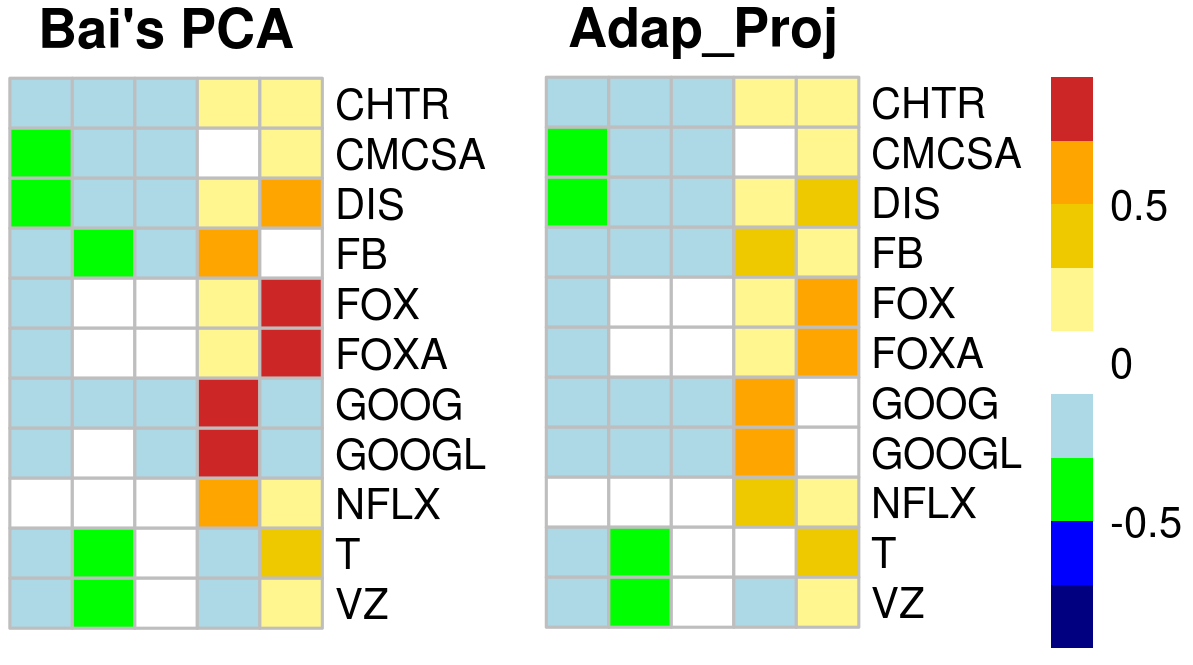}
	\caption{CommunicationServices}
\end{subfigure}
\begin{subfigure}{.3\textwidth}
	\centering
	\includegraphics[width=5cm,height=3cm]{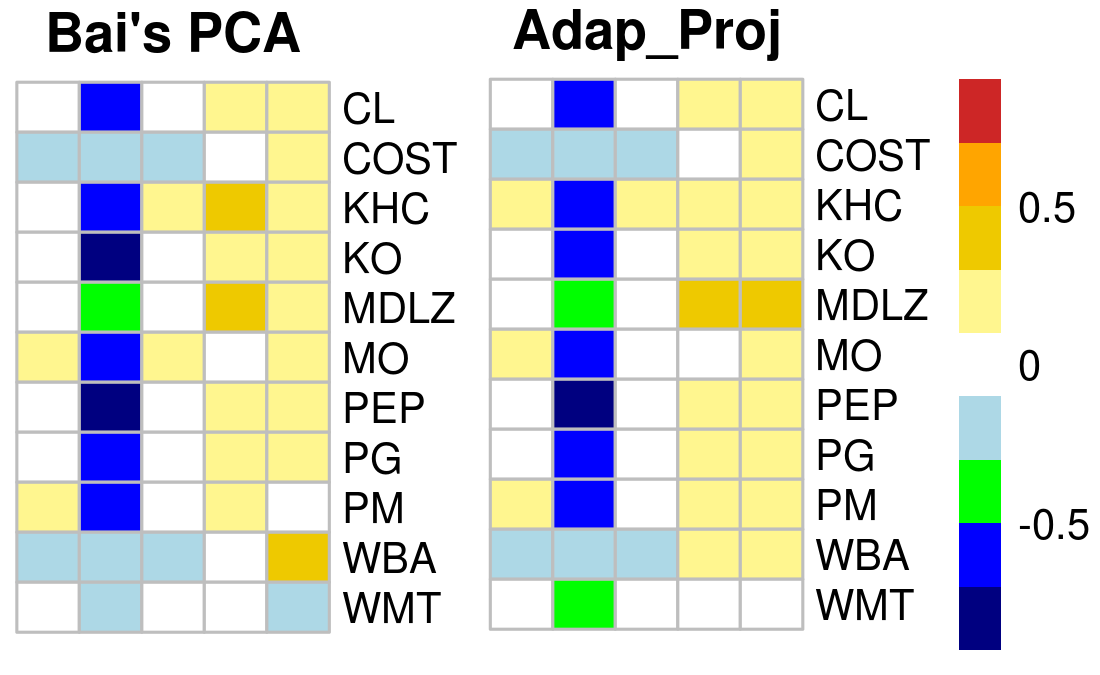}
	\caption{ConsumerStaples}
\end{subfigure}
	\begin{subfigure}{.3\textwidth}
		\centering
		\includegraphics[width=5cm,height=3cm]{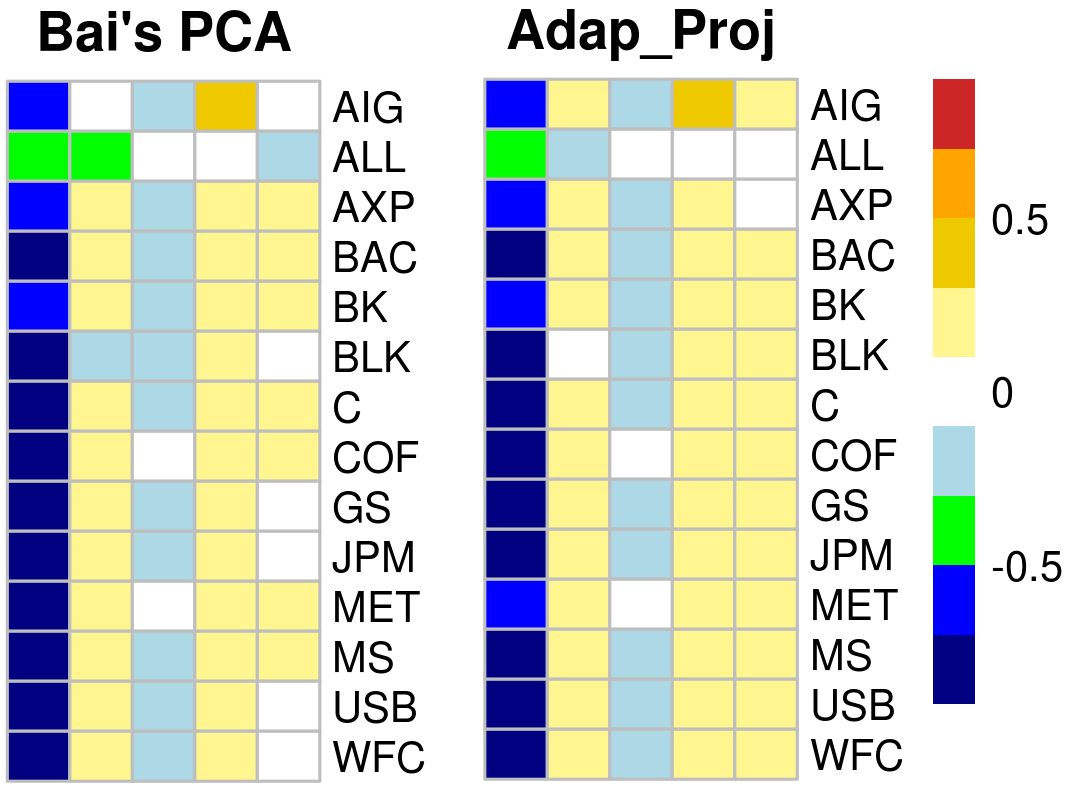}
		\caption{Financials}
	\end{subfigure}

	\begin{subfigure}{.3\textwidth}
		\centering
		\includegraphics[width=5cm,height=3cm]{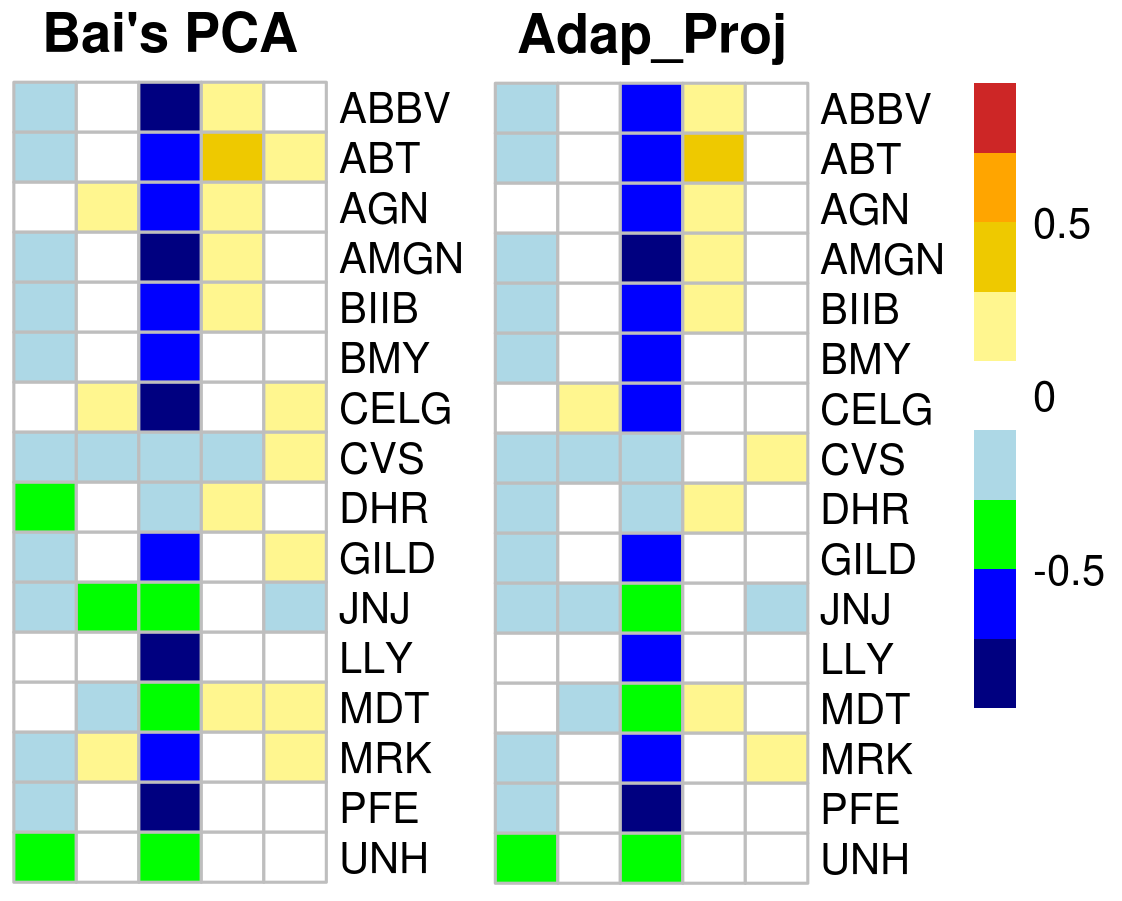}
		\caption{HealthCare}
	\end{subfigure}
	\begin{subfigure}{.3\textwidth}
		\centering
		\includegraphics[width=5cm,height=3cm]{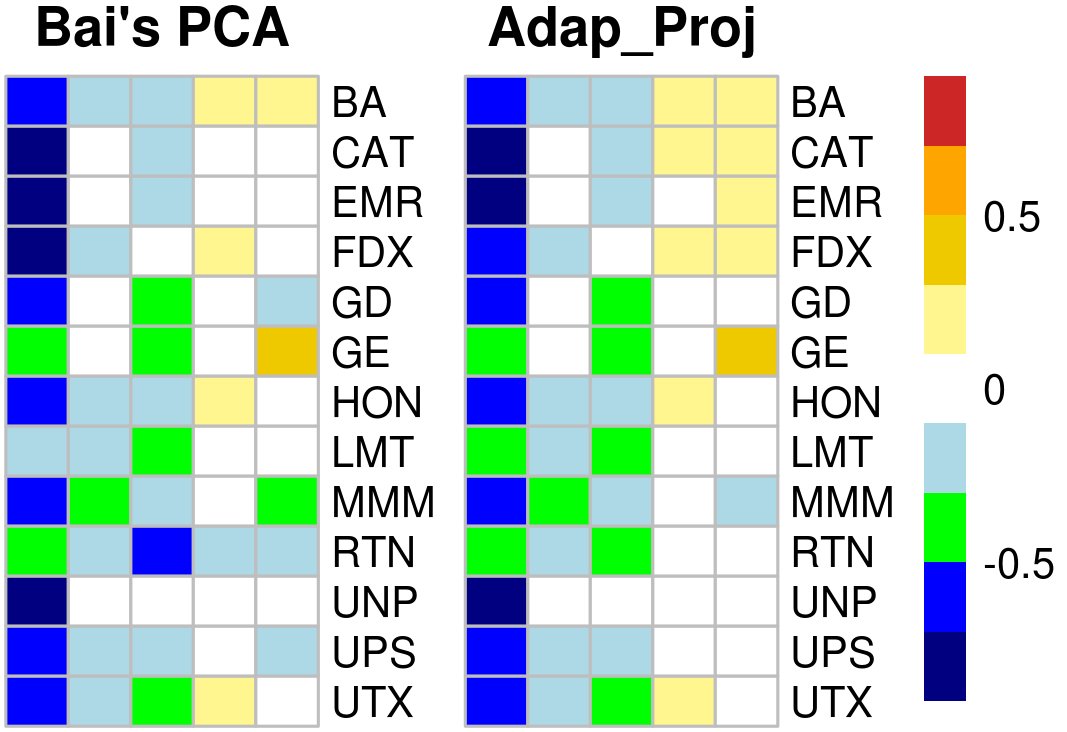}
		\caption{Industrials}
	\end{subfigure}
	\begin{subfigure}{.3\textwidth}
		\centering
		\includegraphics[width=5cm,height=3cm]{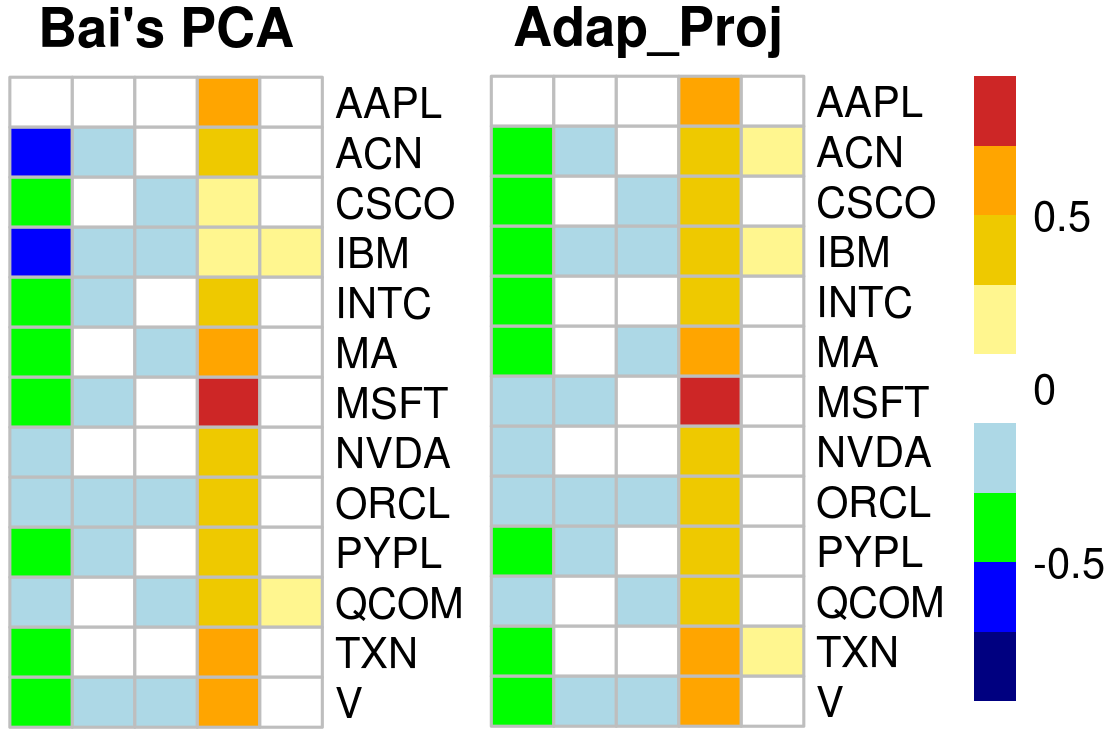}
		\caption{InformationTechnology}
	\end{subfigure}

	\begin{subfigure}{.3\textwidth}
	\centering
	\includegraphics[width=5cm,height=1.5cm]{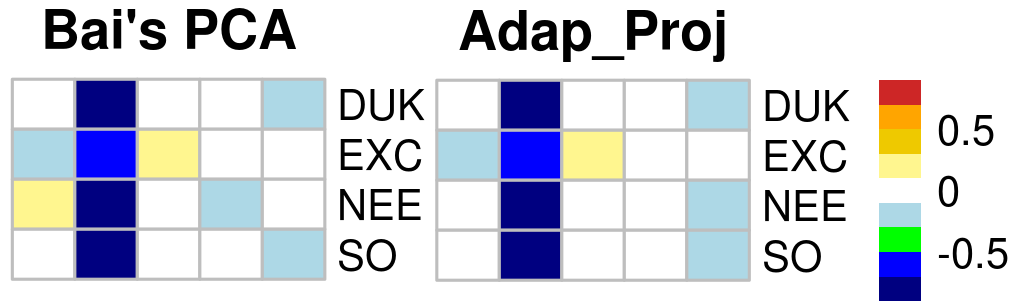}
	\caption{Utilities}
\end{subfigure}
\end{figure}
	\end{appendices}

\end{document}